\documentclass[a4paper,12pt,onecolumn]{article}

\usepackage{amsthm}
\usepackage{graphics} % for pdf, bitmapped graphics files

\usepackage{epsfig} % for postscript graphics files

\usepackage{mathptmx} % assumes new font selection scheme installed

\usepackage{times} % assumes new font selection scheme installed
 
\usepackage{amsmath} % assumes amsmath package installed

\usepackage{amssymb}  % assumes amsmath package installed

\usepackage{amsthm}

\usepackage{color}

\usepackage{algorithm}

\usepackage{algorithmic}

\usepackage{enumerate}

\usepackage{caption}

\newtheorem{thm}{Theorem}
\newtheorem{prop}{Proposition}

\newtheorem{lem}{Lemma}
\newtheorem{defn}{Definition}
\newtheorem{rem}{Remark}

\newcommand{\RR}{{\mathbb{R}}}
\newcommand{\NN}{{\mathbb N}}

\newcommand{\norm}[1]{\lVert#1\rVert}

\newcommand{\dotex}{{\frac{d}{dt}}}

\newcommand{\Rg}{{\mathbb R^{\text{dim } \mathfrak g}}}

\begin{document}

\title{The Invariant Extended Kalman filter as a stable observer}
\author{Axel Barrau, Silv\` ere Bonnabel
\thanks{A. Barrau and S. Bonnabel are with MINES
ParisTech, PSL Research University, Centre for
robotics, 60 Bd St Michel 75006 Paris, France.{\tt\small [axel.barrau,silvere.bonnabel]@mines-paristech.fr}} }
\date{}

\maketitle

\thispagestyle{empty}

\pagestyle{empty}

\begin{abstract}
 We analyze the convergence aspects of the invariant extended Kalman filter (IEKF), when the latter is used as a deterministic non-linear observer on Lie groups, for continuous-time systems with discrete observations. One of the main features of invariant observers for left-invariant systems on Lie groups is that the estimation error is autonomous. In this paper we first generalize this result by characterizing the (much broader) class of systems  for which this property holds. Then, we leverage the result to prove for those systems the local stability of the IEKF around \emph{any} trajectory, under the standard conditions of the linear case. One mobile robotics example and one inertial navigation example illustrate the interest of the approach. Simulations evidence the fact that the EKF is capable of diverging in some challenging situations, where the IEKF with identical tuning keeps converging. 
\end{abstract}

%%%%%%%%%%%%%%%%%%%%%%%%%%%%%%%%%%%%%%%%%%%%%%
% INTRODUCTION
%%%%%%%%%%%%%%%%%%%%%%%%%%%%%%%%%%%%%%%%%%%%%%

\section{Introduction}

The design of non-linear observers is always a challenge, as except for a few classes of systems  (e.g.,  \cite{Gauthier-Kupka-book01}), no general method exists. Of course, the grail of non-linear observer design is to achieve global convergence  to zero of the state estimation error, but this is a very ambitious property to pursue. As a first step, a general method is to use standard linearization techniques, such as the extended Kalman filter (EKF) that makes use of Kalman equations to stabilize the linearized estimation error, and then attempt to derive local convergence properties around \emph{any} trajectory. This is yet a rare property to obtain in a non-linear setting (see, e.g., \cite{aghannan-rouchon-ieee03}), due to the fact that the linearized estimation error equation is time varying, and contrarily to the linear case it generally depends on the unknown true state we seek to estimate.  The EKF,  the most popular observer in the engineering world, provides an ``off the shelf" candidate observer, potentially able to deal with the time-varying nature of the linearized error equation, due to its adaptive gain tuning through a Riccati equation.  However, the EKF does not possess any optimality guarantee, and its efficiency is aleatory. Indeed, its main flaw lies in its very nature: the Kalman gain is computed assuming the estimation error is sufficiently small to be propagated analytically through a first-order linearization of the dynamics about the \emph{estimated} trajectory. When the estimate is actually far from the true state variable, the linearization is not valid, and results in an unadapted gain that may amplify the error. In turn, such positive feedback loop may lead to divergence of the filter. This is the reason why most of the papers dealing with the stability of the EKF (see \cite{boutayeb,song-grizzle-95,reif,bonnabel2012contraction}) rely on the highly non-trivial assumption that the eigenvalues of the Kalman covariance matrix $P_t$ computed about the \emph{estimated} trajectory are lower and upper bounded by strictly positive scalars.  To the authors' knowledge, only a few papers deal with the stability of the EKF without invoking this assumption \cite{krener2003convergence}. It is then replaced by second-order properties whose verification can prove difficult in practical situations. This lack of guarantee is also due to the fact the filter \emph{can} diverge indeed in a number of applications. Note that,  beyond the general theory, there are not even that many engineering examples where the EKF is proved to be (locally) stable. 

The present paper builds upon the theory of symmetry preserving observers  \cite{bonnabel2008symmetry,bonnabel2009non} and notably the theory of invariant Kalman filtering \cite{bonnabel2007left,bonnabel2009invariant,martin2010generalized,barrau2013intrinsic}  in a purely \emph{deterministic} context. As such, it is a contribution to the theory of non-linear observer design on Lie groups that has lately attracted considerable interest, notably for attitude estimation,  see, e.g.,  \cite{mahony2005complementary,Vasconcelos,barczyk2013invariant,
hua2010attitude,lageman2010gradient,grip2015globally,izadi2014rigid}. The detailed contributions and organization of the paper are as follows.

In Section \ref{sect::deterministic}, we recall the main contribution of \cite{bonnabel2008symmetry,bonnabel2009non} is to evidence the fact that for left-invariant systems on Lie groups, non-linear observers may be designed in such a way that the left-invariant estimation error obeys an autonomous equation, a key property for observer design on Lie groups. We show here this property of the error equation can actually be obtained for a \emph{much} broader class of systems, and we characterize this class. Very surprisingly, it turns out that, up to a suitable non-linear change of variables, the error evolution (in the absence of measurements) obeys a \emph{linear} differential equation.

In Section    \ref{sect::IEKF}, we focus on the invariant extended Kalman filter (IEKF) \cite{bonnabel2007left} when applied to the broad class of systems of Section \ref{sect::deterministic}. We consider continuous-time models with discrete observations, which best suits navigation systems where high rate  sensors governing the dynamics are to be combined with low rate sensors \cite{farrell2008aided}. We  change a little the IEKF equations to cast them into a matrix Lie group framework, more handy to use than   the usual abstract Lie group formulation of \cite{bonnabel2007left}. We then prove, that under the standard convergence conditions of the linear case \cite{deyst}, applied to   the linearized model around the \emph{true} state, the IEKF is an asymptotic observer around \emph{any} trajectory of the system, a rare to obtain property. This way, we produce a generic observer with guaranteed local convergence properties under natural assumptions, for a broad class of systems on Lie groups, whereas this property has so far only been reserved to specific examples on Lie groups.   This also allows putting on firm theoretical ground the good behavior of the IEKF in practice, as already noticed in a few papers, see e.g.,  \cite{barczyk2013invariant,barrau2013intrinsic,barczyk2015invariant}.

In Section \ref{sect::examples:A} we consider a mobile robotics example, where a unicycle robot (or simplified car) tries to estimate its position and orientation from GPS position (only) measurements, or alternatively landmarks range and bearing measurements. On this example of engineering interest, the IEKF is proved to converge around \emph{any}  trajectory using the results of the paper, which is a contribution in itself. Simulations indicate the IEKF is always superior to the EKF and may even outperform the latter in challenging situations.  

In Section \ref{sect::examples:B} we consider the highly relevant problem of an unmanned aerial vehicle (UAV) navigating with accelerometers and gyrometers, and range and bearing measurements of known landmarks. Although the system is not invariant in the sense of \cite{bonnabel2008symmetry,bonnabel2009non}, it is proved to fit into our framework so that the  autonomous error equation property of \cite{bonnabel2009non}  holds, a fact never  noticed before to our best knowledge (except in our preliminary conference paper \cite{barrau-bonnabel-cdc14}). The IEKF is shown to converge around \emph{any} trajectory using the results of the paper, which is a contribution in itself. Moreover, it is shown to outperform the EKF which even  diverges when, as in high precision navigation, the user has way more trust in the inertial sensors  than in the landmark measurements.

The  main contributions can be summarized as follows:
\begin{itemize}
\item The class of systems, for which the key result of \cite{bonnabel2009non} about the (state) error equation autonomy holds, is completely characterized, and actually shown to be much broader than left-invariant systems. 
\item The autonomy of the error equation is proved to come with a very intriguing property: a well-chosen non-linear function of the non-linear error is proved to obey a linear differential equation.
\item In turn, this property allows proving that, for the introduced class of systems, the IEKF used in a deterministic context possesses powerful local convergence guarantees that the standard EKF lacks. 
 \item Two examples of navigation illustrate the applicability of the results, and simulations indicate indeed the IEKF is always superior to the EKF, and may turn out to literally outperform the latter when confronted with some challenging situations - the EKF being even capable to diverge. 
\end{itemize}

%Preliminary work can be found in \cite{bonnabel2007left,bonnabel2009invariant, barrau-bonnabel-cdc13}.

%%%%%%%%%%%%%%%%%%%%%%%%%%%%%%%%%%%%%%%%%%%%%%
% SECTION 1
%%%%%%%%%%%%%%%%%%%%%%%%%%%%%%%%%%%%%%%%%%%%%%

\section{A special class of multiplicative systems}
\label{sect::deterministic}

\subsection{An introductory example}\label{introd:ex}

Consider a linear (deterministic) system $\dotex x_t=A_tx_t$. Consider two trajectories of this system, say, a reference trajectory  $(\bar x_t)_{t\geq 0}$ and another one $(x_t)_{t\geq 0}$. The  discrepancy between both trajectories $\Delta x_t:=x_t-\bar x_t$ satisfies the linear equation $\dotex \Delta x_t=A_t \Delta x_t$.  This is a key property for the design of linear convergent observers, as during the propagation step, the evolution of the error between the true state and the estimate does \emph{not} depend on the true state's trajectory.

Consider now the following non-linear standard  model of the two-dimensional non-holonomic car. Its state is defined by three parameters : heading $\theta_t$ and position $X_t=(x_t^1,x_t^2)$. The velocity $u_t\in\RR$ is given by an odometer, the angular velocity $\omega_t\in\RR$ is measured by a differential odometer or a gyrometer. The equations read (see, e.g., \cite{de1998feedback}):
$$
 \frac{d}{dt} \theta_t  = \omega_t , \qquad \frac{d}{dt} x_t^1  = \cos(\theta_t) u_t , \qquad 
 \frac{d}{dt} x_t^2  = \sin(\theta_t) u_t. 
$$
Now consider a reference trajectory $(\bar{\theta}_t,\bar{X}_t)$ and a second trajectory $(\theta_t,X_t)$ with different initial conditions but same inputs. The exact propagation of the ``error" $(\Delta \theta_t, \Delta X_t)=(\theta_t - \bar{\theta_t}$, $ X_t-\bar{X_t})$,  satisfies:
\begin{equation}\begin{aligned}\label{ouaich}
 &\frac{d}{dt} \Delta \theta_t  = 0 , \\
& \frac{d}{dt} \Delta x_t^1  = \left[ \cos( \theta_t) - \cos(\bar{\theta_t})  \right] u_t ,\\&
 \frac{d}{dt} \Delta x_t^2  = \left[ \sin(\theta_t) - \sin(\bar{\theta_t})  \right] u_t ,
\end{aligned}\end{equation}
where we let $\Delta X_t=(\Delta x_t^1 ,\Delta x_t^2)$. We see the time derivative of  $(\Delta \theta_t, \Delta X_t)$ is not a function of $(\Delta \theta_t, \Delta X_t)$ only: it also involves $\bar \theta_t$ and $\theta_t$ individually. Moreover, the equation is non-linear. These features, characteristic of non-linear systems, make the design of observers way more complicated in the non-linear case. Now, let us introduce the following non-linear error, where $R(\theta) = \cos(\theta)I_2 + \sin(\theta)J_2 $ denotes the planar rotation matrix of angle $\theta$  (see definition of $J_2$ below):
\begin{equation}\begin{aligned}\label{xi:ccar}
\xi_t:= & \begin{pmatrix} (\theta_t - \bar{\theta_t}) \\ \frac{1}{2} \left( \theta_t - \bar \theta_t \right) \left[ a(\theta_t - \bar \theta_t ) I_2 - J_2 \right] R(-\bar{\theta_t})(X_t-\bar{X_t}) \end{pmatrix}, \\ & \text{ with  } J_2=\begin{pmatrix} 0 & -1 \\ 1 & 0 \end{pmatrix} \text{ and } a(s) = \frac{\sin(s)}{1-\cos(s)},
\end{aligned}\end{equation}
which is equal to 0 indeed if and only if both trajectories coincide. We are about to prove by elementary means a surprising property that will be generalized by Theorem \ref{dechire:thm}.
\begin{prop}Contrarily to the linear error obeying  \eqref{ouaich}, the alternative non-linear error \eqref{xi:ccar} obeys the following  \emph{linear} and autonomous equation although the system, and the error, are totally non-linear:
\begin{equation}
\frac{d}{dt} \xi_t = \begin{pmatrix} 0 & 0 & 0 \\ 0 & 0 & \omega_t \\ -u_t & -\omega_t & 0 \end{pmatrix} \xi_t.
\label{eq::linear_2D}
\end{equation}
\end{prop}
\begin{proof}
We will use the notations   $\delta x_t = \frac{1}{2} \left( \theta_t - \bar \theta_t \right) \left[ a(\theta_t - \bar \theta_t ) I_2 - J_2 \right] R(-\bar{\theta_t})(X_t-\bar{X_t})$. Since we have  $\frac{d}{dt} \delta \theta_t=0$ as in the linear case above, only the two last terms of $\delta x_t$ change over time. Moreover, $J_2$ commutes with $I_2$, $J_2$ and $R(-\bar \theta_t)$ and $\dotex R(-\bar\theta_t)=-J_2\omega_tR(-\bar\theta_t)$. Thus we have:
\begin{align*}
 \frac{d}{dt} & \delta x_t = -J_2 \omega_t \delta x_t \\
 & +  \frac{1}{2} \delta \theta_t \left[ a(\delta \theta_t ) I_2 - J_2 \right] R(-\bar{\theta_t}) \left[ R(\theta_t) -R( \bar \theta_t) \right] (u_t , 0)^T \\
 = & -J_2 \omega_t \delta x_t \\
 & +  \frac{1}{2} \delta \theta_t \left[ a(\delta \theta_t ) I_2 - J_2 \right] \left[ (\cos(\delta \theta_t)-1) I_2 + \sin(\delta \theta_t) J_2  \right] (u_t , 0)^T \\
 = & - J_2\omega_t  \delta x_t +   \delta \theta_t J_2 (u_t , 0)^T \\
 = & -(\omega_t J_2)  \delta x_t +   \delta \theta_t (0 , -u_t)^T.
\end{align*}
In the second to last equality we used the relation $\frac{\sin(s)^2}{1-\cos(s)} = 1+ \cos(s) $. Equation \eqref{eq::linear_2D} is proved.
\end{proof}
 The present section provides a novel geometrical framework - encompassing this example - to characterize systems on Lie groups for which such a property holds. In turn, such a property will simplify the convergence analysis of non-linear observers, namely the IEKF,  due to the implied similarities with the linear case.

\subsection{Systems on Lie groups with state trajectory independent error propagation property}

Let $G\subset \RR^{N\times N}$ be a matrix Lie group whose Lie algebra is denoted $\frak{g}$ and has dimension $\mathrm{dim}~\frak{g}$. We consider a class of dynamical systems:
\begin{align}
\frac{d}{dt}\chi_t & = f_{u_t}(\chi_t), \label{eq:IEKF_model}
\end{align}
where the state $\chi_t$ lives in the Lie group $G$ and ${u_t}$ is an input variable. Consider two distinct trajectories $\chi_t$ and $\bar\chi_t$ of \eqref{eq:IEKF_model}.  Define the left-invariant and right-invariant errors  $\eta_t^L$ and $\eta_t^R$ between the two trajectories as:
\begin{equation}
\eta_t^L=\chi_t^{-1} \bar{\chi}_t \qquad \text{(left invariant),}\label{leftinv}
\end{equation}
\begin{equation}
\eta_t^R=\bar{\chi}_t \chi_t^{-1} \qquad \text{(right invariant).}\label{rightinv}
\end{equation}
The  terminology stems from the invariance of e.g.,  \eqref{leftinv} to (left) multiplications $(\chi,\bar\chi)\to(\Gamma\chi,\Gamma\bar\chi)$ for $\Gamma\in G$. 
\begin{defn}
The left-invariant and right-invariant errors are said to have a  state-trajectory independent propagation if they satisfy a differential equation of the form $\frac{d}{dt}\eta_t = g_{u_t}(\eta_t)$.
\end{defn}

Note that, in general the time derivative of $\eta_t$ is a complicated function depending on $u_t$ and both $\chi_t$ and $\bar{\chi}_t$ in a way that does not boil down to a function of $\eta_t$, see for instance eq \eqref{ouaich} above. The following result allows characterizing the class of systems of the form \eqref{eq:IEKF_model} for which the property holds. 

\begin{thm}
\label{thm::CNS}
 The three following conditions are equivalent for the dynamics \eqref{eq:IEKF_model}:
\begin{enumerate}[i]
\item The left-invariant error \eqref{leftinv} is state trajectory independent
\item The right-invariant error \eqref{rightinv} is state trajectory independent
\item For all $t>0$ and $a,b \in G$ we have (in the tangent space at $ab$):
\begin{equation}
\label{eq::main_relation}
f_{u_t}(ab) = f_{u_t}(a)b + af_{u_t}(b) - a f_{u_t}(I_d)b,
\end{equation}
\end{enumerate}
where $I_d$ denotes the identity matrix. Moreover, if one of these conditions is satisfied we have 
\begin{align}\frac{d}{dt}\eta_t^L &= g_{u_t}^L(\eta_t^L)\quad\text{where}\quad g_{u_t}^L(\eta)=f_{u_t}(\eta)-f_{u_t}(I_d)\eta \label{g_f_1},\\\frac{d}{dt}\eta_t^R &= g_{u_t}^R(\eta_t^R)\quad\text{where}\quad
g_{u_t}^R(\eta)=f_{u_t}(\eta)-\eta f_{u_t}(I_d). \label{g_f_2}\end{align}
\end{thm}

\begin{proof}
Assume we have $\frac{d}{dt}\eta_t^L = g_{u_t}(\eta_t^L)$ for a certain function $g_{u_t}$ and any  $\eta_t^L = \chi_t^{-1} \bar{\chi}_t$, where $\chi_t$ and $\bar{\chi}_t$ are solutions of \eqref{eq:IEKF_model}. We have:
\begin{align}
&g_{u_t}(\chi_t^{-1} \bar{\chi}_t)  = \frac{d}{dt} (\chi_t^{-1} \bar{\chi}_t) = -\chi_t^{-1} [\frac{d}{dt} \chi_t] \chi_t^{-1} \bar{\chi}_t + \chi_t^{-1} \frac{d}{dt}\bar{\chi}_t \nonumber\\&\qquad = -\chi_t^{-1} f_{u_t}(\chi_t)  \eta_t^L + \chi_t^{-1}f_{u_t}(\bar{\chi}_t), \nonumber \\
&\text{i.e. } g_{u_t}(\eta_t^L)  = -\chi_t^{-1} f_{u_t}(\chi_t) \eta_t^L + \chi_t^{-1} f_{u_t}(\chi_t \eta_t^L) \label{g_f_11}.
\end{align}
This has to hold for any $\chi_t$ and $\eta_t^L$. In the particular case where $\chi_t = I_d$ we obtain:
\begin{equation}
g_{u_t}(\eta_t^L)=f_{u_t}(\eta_t^L)-f_{u_t}(I_d)\eta_t^L . \label{g_f_22}
\end{equation}
Reinjecting \eqref{g_f_22} in \eqref{g_f_11} we obtain:
\[
f_{u_t}(\chi_t \eta_t^L) = f_{u_t}(\chi_t)\eta_t^L + \chi_tf_{u_t}(\eta_t^L) - \chi_t f_{u_t}(I_d) \eta_t^L.
\]
The converse is trivial and the proof is analogous for right-invariant errors.
\end{proof}

\begin{rem}
The particular cases of left-invariant and right-invariant dynamics, or the combination of both as follows, verify \eqref{eq::main_relation}. Let $f_{v_t,\omega_t} (\chi) = v_t \chi +  \chi \omega_t$. We have indeed:
\begin{align*}
&f_{v_t,\omega_t}(a)b+af_{v_t,\omega_t}(b)- a f_{v_t,\omega_t}(I_d)b \\&=  (v_t a + a \omega_t) b + a(v_t b + b \omega_t) - a (v_t + \omega_t) b \\
 &  =  v_t a b + ab \omega_t = f_{v_t,\omega_t}(ab).
\end{align*}
\end{rem}

\begin{rem}
In the particular case where $G$ is a vector space with  standard addition as the group composition law,  the condition \eqref{eq::main_relation} boils down to $f_{u_t}(a+b) = f_{u_t}(a) + f_{u_t}(b) -f_{u_t}(0)$ and we recover the affine functions. We thus see the class of system introduced here  appears as a generalization of the linear case. 
\end{rem}

%\begin{align}
%\frac{d}{dt}\chi_t & = f_{u_t}(\chi_t)+\chi_t w_t \label{eq::IEKF_model} \\
%Y_{t_n} & = \psi(\chi_{t_n} d)+\begin{align}

%\end{align}
%where $\chi\in\RR^{m\times m}$, $d\in\mathcal Y\subset \RR^m$ and $\psi:\RR^m\supset\mathcal Y\to\RR^p$ is an invertible projection that extracts some coordinates. 
%
%We can also write $\chi_t w_t=W_t\chi_t$ with $W_t=\chi_t w_t\chi_t^{-1}$. We can do the same with $V_n$. 
%
%\vspace{.3cm}
%
%
%IEKF
%\begin{align}
%\frac{d}{dt}\hat{\chi}_t & = f_{u_t}(\hat{\chi}_t) \label{IKF_propagation} \\
%\hat{\chi}_{t_n}^+ & = \hat{\chi}_{t_n} \phi[K_n(\hat{\chi}_{t_n}^{-1} \psi^{-1}(Y_{t_n})- d)] \label{IKF_update}
%\end{align}
%
%The associated invariant error equation reads:
%\begin{align}
%\frac{d}{dt} \eta_t & = f_{u_t}(\eta_t) - f_{u_t}(I_d)\eta_t - w_t \eta_t \label{error_propagation} \\
%\eta_{t_n}^+ & = \eta_{t_n} \phi[K_n(\eta_{t_n} \cdot d - d + \hat{\chi}_{t_n}^{-1}\psi^{-1}(V_n))] \label{error_update}
%\end{align}
%
%We have the result
%\begin{thm}
%\label{thm::CNS}
%The evolution of the invariant error variable $\eta_t = \chi^{-1}_t \hat{\chi}_t$ is state trajectory independent, i.e. verifies an equation of the form $\frac{d}{dt}\eta_t = g_{u_t}(\eta_t)$, under the following necessary and sufficient condition: $
%\forall a,b \in G, f_{u_t}(ab) = f_{u_t}(a)b + af_{u_t}(b) - a f_{u_t}(I_d)b $. Moreover $
%g_{u_t}(\eta)=f_{u_t}(\eta)-f_{u_t}(I_d)\eta \label{g_f_2}$.
%\end{thm}
In the next section we show that the dynamics of the form \eqref{eq:IEKF_model} with additional property  \eqref{eq::main_relation} have striking properties  generalizing those of linear systems.

%%%%%%%%%%%%%%%%%%%%%%%%%%%%%%%%%%%%%%%%%%%%%%%%%%%%
%%%%%%%%%%%%%%%%%%%%%%%%%%%%%%%%%%%%%%%%%%%%%%%%%%%%

\subsection{Log-linear property of the error propagation}
In the sequel, we will systematically consider systems of the form \eqref{eq:IEKF_model} with the additional property  \eqref{eq::main_relation}, i.e.  systems on Lie groups defined by
\begin{equation}
\label{eq::mult_linear}
\begin{aligned}\frac{d}{dt} \chi_t & = f_{u_t}(\chi_t), \\
\quad\text{where $\forall (u,a,b)$}\quad
f_u(ab) & =af_u(b)+f_u(a)b-af_u(Id)b.
\end{aligned}
\end{equation}

%%%%%%%%%%%%%%%%%%%%%%%%%%%%%%%%%%%%%%%%%%%%%%%%%%%%%%
% THEOREM
%%%%%%%%%%%%%%%%%%%%%%%%%%%%%%%%%%%%%%%%%%%%%%%%%%%%%%

For such systems, Theorem \ref{thm::CNS} proves that the left (resp. right) invariant error  is a solution to the equation  $\frac{d}{dt}\eta_t=g_{u_t}(\eta_t)$ where $g_{u_t}$ is given by   \eqref{g_f_1} (resp. \eqref{g_f_2}).  We have the following novel and striking property.

\begin{thm}\label{dechire:thm}[Log-linear property of the error]
Consider the left or right invariant error $\eta_t^i$ as defined by \eqref{leftinv} or \eqref{rightinv} between two arbitrarily far  trajectories of \eqref{eq::mult_linear}, the superscript  $i$ denoting indifferently $L$ or $R$. Let $\mathcal{L}_{\mathfrak{g}}$ and $\exp(.)$ be defined as in Appendix \ref{sect::tuto_Lie_groups}. Let $\xi_0^i\in\RR^{\text{dim }\mathfrak  g}$ be such that initially $\exp(\xi_0^i)=\eta_0^i$. Let $A_{u_t}^i$ be  defined by $g_{u_t}(\exp(\xi)) = \mathcal{L}_{\mathfrak{g}}(A_{u_t}^i \xi) + O(\norm{\xi}^2)$. If $\xi_t^i$ is defined for $t>0$ by the \emph{linear} differential equation in  $\RR^{\text{dim } \mathfrak g}$
\begin{equation}
\label{eq::def_xi}
\dotex \xi_t^i=A_t^i\xi_t^i,
\end{equation}
\emph{then}, we have for the true non-linear error $\eta_t^i$, the correspondence \emph{at all times} and for arbitrarily large errors 
$$
\forall t\geq 0\quad \eta_t^i=\exp(\xi_t^i).
$$
 \end{thm}The latter result, whose proof has been moved to the Appendix, shows that a wide range of nonlinear problems (see examples below) can lead to linear error equations provided the error variable is correctly chosen.  We also see the results displayed in the previous introductory example of Section \ref{introd:ex} are mere applications of the latter theorem, as the non-holonomic car example turns out to perfectly fit into our framework (see Section \ref{sect::examples:A}) and  $\xi_t$ in eq \eqref{xi:ccar} actually  merely is the Lie logarithm of the left-invariant error.   
 This will be extensively used in Section \ref{sect::IEKF}, and in the examples to prove stability properties of IEKFs.

%%%%%%%%%%%%%%%%%%%%%%%%%%%%%%%%%%%%%%%%%%%%%%%%%
% IEKF
%%%%%%%%%%%%%%%%%%%%%%%%%%%%%%%%%%%%%%%%%%%%%%%%%

\section{Invariant Extended Kalman Filtering}
\label{sect::IEKF}

In this section we first recap the equations of the Invariant EKF (IEKF), a variant of the EKF devoted to Lie groups space states, that has been introduced in continuous time in \cite{bonnabel2007left,bonnabel2009invariant}. We derive the equations in continuous time with discrete observations here, which has already been done in a restricted setting in \cite{barrau2013intrinsic}, and we propose a novel matrix (Lie group) framework to simplify the design. We then show that for the class of systems introduced in Section \ref{sect::deterministic}, under observability conditions, and painless conditions on the covariance matrices considered here as design parameters, the IEKF is a (deterministic) non-linear observer with local convergence properties around \emph{any} trajectory, a feature extremely rare to obtain in the field of non-linear observers, due to the dependency of the estimation error to the true unknown trajectory. The notions necessary to follow Section \ref{sect::IEKF} are given in Appendix \ref{sect::tuto_Lie_groups}. 

\subsection{Full system and IEKF general structure}

We consider in this section an equation on a matrix Lie group $G\subset \RR^{N\times N}$ of the form:
\begin{align}
\frac{d}{dt}\chi_t & = f_{u_t}(\chi_t), \label{eq::IEKF_modelle}
\end{align}
with $
f_u(ab)=af_u(b)+f_u(a)b-af_u(Id)b$ for all $(u,a,b)\in U\times G\times G$. This system will be associated to two different kinds of observations.

\subsubsection{Left-invariant observations}
The first family of outputs we are interested in write:
\begin{align}\label{claude:eq}
Y^1_{t_n} & = \chi_{t_n} d^1  \qquad , \qquad
  ... \qquad , \qquad
Y^k_{t_n} = \chi_{t_n} d^k,
\end{align}
where $(d^i)_{i \leq k}$ are known vectors.  The Left-Invariant Extended Kalman Filter (LIEKF) is defined in this setting through the following propagation  and update steps:
\begin{align}
\frac{d}{dt}\hat{\chi}_t & = f_{u_t}(\hat{\chi}_t),\quad  t_{n-1}\leq t< t_n , \qquad\quad\text{Propagation} \label{IEKF_propagation} \\
\hat{\chi}_{t_n}^+ & = \hat{\chi}_{t_n} \exp \left[ L_n
\begin{pmatrix}
\hat{\chi}_{t_n}^{-1} Y_{t_n}^1- d^1 \\
... \\
\hat{\chi}_{t_n}^{-1} Y_{t_n}^k - d^k
\end{pmatrix} \right], \quad\text{Update} \label{LIEKF_update}
\end{align}where the function $L_n:\RR^{kN}\to \RR^{\dim \mathfrak{g}}$ is to be defined in the sequel using error linearizations. 
A left-invariant error between true state $\chi_t$ and estimated state $\hat \chi_t$ can be associated to this filter:
\begin{align}\label{tule}
\eta_t^L = \chi_t^{-1} \hat{\chi}_t.
\end{align}
During the Propagation step,  ${\chi}_t$ and $\hat{\chi}_t$ are two trajectories of the system \eqref{eq::IEKF_modelle}. Thus, the error \eqref{tule} is independent from the true state trajectory from Theorem \ref{thm::CNS} and eq. \eqref{g_f_1} ! We have thus
\begin{align}
\dotex\eta_t^L=g_{u_t}^L(\eta_t^L),\quad t_{n-1}\leq t<t_n.
\label{er_r:eqrev}
\end{align}Consider now the following linear differential equation in $\RR^{\text{dim }\mathfrak{g}}$: 
\begin{align}
\frac{d}{dt} \xi_t = A_{u_t}\xi_t ,  \label{line:eqrev}
\end{align}
where $A_{u_t}$ is defined by $g_{u_t}^L(\exp(\xi)) = \mathcal{L}_{\mathfrak{g}}(A_{u_t} \xi) + O(\norm{\xi_t}^2)$. Theorem \ref{dechire:thm} implies the unexpected result:
\begin{prop}\label{mrpror}If $\xi_t$ is defined as a solution to the  \emph{linear} system \eqref{line:eqrev} and  $\eta_t$ is defined as the solution to the \emph{nonlinear} error system \eqref{er_r:eqrev}, then if at time $t_{n-1}$  we have $\eta_{t_{n-1}}=\exp(\xi_{t_{n-1}})$ then the equality $\eta_t\label{expmap}=\exp(\xi_t)=\exp_m \left( \mathcal{L}_{\mathfrak{g}}({\xi}_t) \right)
$ is verified at all times $ t_{n-1}\leq t<t_n$, even for arbitrarily large initial errors.\end{prop}

Besides, at the update step, the evolution of the invariant error variable \eqref{tule} merely writes:
\begin{align}
(\eta_{t_n}^L)^+ = \chi_{t_n}^{-1} \hat{\chi}_{t_n}^+
= \eta_{t_n}^L \exp \left[ L_n
\begin{pmatrix}
(\eta_{t_n}^L)^{-1} d^1 -  d^1  \\
... \\
(\eta_{t_n}^L)^{-1} d^k -  d^k 
\end{pmatrix} \right].\label{update:evolution:left:eq}
\end{align}
We see that the nice geometrical structure of the LIEKF allows the updated error  $(\eta_{t_n}^L)^+$ to be here again only a function of the error just before update $\eta_{t_n}^L$, i.e. to be independent from the true state ${\chi}_{t_n}$.

\subsubsection{Right-invariant observations}
The second family of observations we are interested in have the form:
\begin{align}\label{claudio:eq}
Y^1_{t_n} = \chi_{t_n}^{-1 } d^1 \qquad , \qquad
  ... \qquad , \qquad
Y^k_{t_n} = \chi_{t_n}^{-1} d^k,
\end{align}
with the same notations as in the previous section.
The Right-Invariant EKF (RIEKF) is defined here as:
\begin{align}
\frac{d}{dt}\hat{\chi}_t & = f_{u_t}(\hat{\chi}_t),  \label{RIEKF_propp}\\
\hat{\chi}_{t_n}^+ & = \exp \left[ L_n
\begin{pmatrix}
\hat{\chi}_{t_n} Y_{t_n}^1- d^1 \\
... \\
\hat{\chi}_{t_n} Y_{t_n}^k - d^k
\end{pmatrix} \right] \hat{\chi}_{t_n}. \label{RIEKF_update}
\end{align}
A right-invariant error can be associated to this filter:
\begin{align}\label{skypi}
\eta_t^R = \hat{\chi}_t \chi_t^{-1}.
\end{align}
Once again,  Theorem \ref{thm::CNS} proves the evolution of the error does not depend on the state of the system. The analog of Proposition \ref{mrpror} is thus easily derived for the error \eqref{skypi} and we skip it due to space limitations.

At the update step, the evolution of the invariant error variable reads:
\[
(\eta_{t_n}^R)^+ =\hat{\chi}_t^+ \chi_t^{-1}
= \exp \left[ L_n
\begin{pmatrix}
\eta_{t_n}^R d^1 -  d^1 \\
... \\
\eta_{t_n}^R d^k -  d^k
\end{pmatrix} \right] \eta_{t_n}^R,
\]
so that the error update does not depend on the true state either.

\subsection{IEKF gain tuning}\label{ouida}

In the standard theory of Kalman filtering, EKFs are designed for ``noisy" systems associated with the deterministic considered system. In a deterministic context, the covariance matrices $Q$ and $N$ of the noises are left free to tune by the user, and are design parameters for the EKF used as  a non-linear observer. Yet, in the spirit of \cite{song-grizzle-95}, it is nevertheless convenient to  associate a ``noisy'' system with the considered deterministic system consisting of dynamics \eqref{eq::IEKF_modelle} with outputs \eqref{claude:eq} or \eqref{claudio:eq}. The obtained error equations can be linearized, and the standard Kalman equations applied to make this error decrease. This way, the matrices $Q$ and $N$  can be \emph{interpreted} as covariance matrices. And in many engineering applications, the characteristics of the noises of the sensors are approximately known, so that the engineer can use the corresponding covariance matrices as a useful guide to tune (or design) the  non-linear observer, that is here the IEKF. This provides him with (at least) a first sensible tuning  of the parameter matrices, which is consistent with the trust he has in each sensor. Moreover,  in the same spirit, the IEKF viewed as a non-linear observer remedies a common weakness shared by numerous non-linear observers on Lie groups, as it conveys an information about its own accuracy through the computed covariance matrix $P_t$. Although it comes with no rigorous interpretation in a deterministic context, the information conveyed by $P_t$ may prove useful in applications.

Note that,  in mobile robotics and navigation, the sensors are attached to the earth-fixed frame (e.g., a GPS) or to the body frame (e.g., a gyrometer). To \emph{interpret} them as covariance matrices in the IEKF framework (see below) those matrices $Q$ and $N$ may have to undergo a change of frame yielding trajectory dependent tuning matrices $\hat Q$ and $\hat N$, such as in the application examples in the sequel. This does not weaken the results, but however comes at the price of making the stability analysis a little more complicated. 

%%%%%%%%%

\subsubsection{Associated ``noisy" system}
To tune the IEKF \eqref{IEKF_propagation}-\eqref{LIEKF_update} or  \eqref{RIEKF_propp}-\eqref{RIEKF_update}, we associate to the system \eqref{eq::IEKF_modelle} the following
``noisy" system:
\begin{align}
\frac{d}{dt}\chi_t & = f_{u_t}(\chi_t)+\chi_t w_t, \label{eq::IEKF_model}
\end{align}
 where
 $w_t$ is a continuous white noise belonging to $\frak{g}$ whose covariance matrix is denoted by $Q_t$ (for a proper discussion on multiplicative noise for systems defined on Lie groups, see e.g. \cite{barrau2013intrinsic}).

In the same way, we associate to the family of left-invariant observations \eqref{claude:eq} the following family of ``noisy" outputs:
\begin{align}\label{claude:noise}
Y^1_{t_n} & = \chi_{t_n} \left( d^1+B_n^1 \right) + V_n^1 ~ , ~
  ... ~ , ~
Y^k_{t_n} = \chi_{t_n} \left( d^k+B_n^k \right) +V_n^k,
\end{align}
where the $(V_n^i)_{i \leq k}$, $(B_n^i)_{i \leq k}$ are noises with known characteristics.  To the family of right-invariant observations \eqref{claudio:eq} we associate the following family of ``noisy" outputs
\begin{align}\label{claudio:noise}
Y^1_{t_n} = \chi_{t_n}^{-1 } \left( d^1+V_n^1 \right) + B_n^1 ~ , ~
  ... ~ , ~
Y^k_{t_n} = \chi_{t_n}^{-1} \left( d^k+V_n^k \right) + B_n^k. \end{align}

\subsubsection{Linearized ``noisy" estimation error equation}
As in a conventional EKF, we assume the error to be small (here close to $Id$ as it is equal to $Id$ if $\hat \chi_t=\chi_t$) so that the error system can be linearized to compute the gains $L_n$. By definition, the Lie algebra $\mathfrak{g}$ represents the infinitesimal variations around $Id$ of an element of $G$. Thus the natural way to define a vector error variable $\xi_t$ in $\Rg$ is (see Appendix \ref{sect::tuto_Lie_groups}):
\begin{align}
\eta_t\label{expmap}=\exp(\xi_t)=\exp_m \left( \mathcal{L}_{\mathfrak{g}}({\xi}_t) \right).
\end{align}

%Since we deal with a matrix Lie group $G$, the usual matrix exponential map allows to map elements from the Lie algebra $\mathfrak{g}$ to the Lie group $G$. The Lie algebra is in turn a vector space of dimension $\mathrm{dim}~\mathfrak{g}$ embedded in $\RR^{N\times N}$. It is isomorphic to $\RR^{\mathrm{dim}\mathfrak {g}}$ via a mapping $\xi \in \RR \to{\check{\xi}}\in \mathfrak{g}$ that will be made explicit in all the examples considered in the sequel. The error can thus be written letting $\xi_t\in \RR^{\mathrm{dim}\mathfrak{g}}$  be defined through the mapping from $\RR^{\mathrm{dim}\mathfrak{g}}$ to $G$

%\begin{align}
%\exp(\check{\xi_t})=\eta_t\label{expmap}
%\end{align}

During the Propagation step, that is for $t_{n-1}\leq t<t_n$, elementary computations based on the results of Theorem \ref{thm::CNS} show that for the noisy model \eqref{eq::IEKF_model} we have
\begin{equation}
\begin{aligned}
\dotex\eta_t^L &  = g_{u_t}^L \left( \eta_t^L \right) - w_t\eta_t^L, \\
 \dotex\eta_t^R & = g_{u_t}^R \left( \eta_t^R \right) - \left( \hat \chi_t w_t \hat \chi_t^{-1} \right) \eta_t^R.
\end{aligned}
\label{er_r:eq}
\end{equation}
Defining $\hat{w}_t \in \RR^{\dim \mathfrak{g}}$ by $\mathcal{L}_{\mathfrak{g}}(\hat{w}_t) = - w_t$ in the first case and $\mathcal{L}_{\mathfrak{g}}(\hat{w}_t) = -\hat\chi_tw_t\hat\chi_t^{-1}$ (i.e. $\hat{w}_t = - Ad_{\hat{\chi}_t} \mathcal{L}_{\mathfrak{g}}^{-1} (w_t)$) in the second case, and using the superscript $i$ to denote indifferently $L$ or $R$ we end up with the linearized error equation in $\RR^{\mathrm{dim }\mathfrak{g}}$: 
\begin{align}
\frac{d}{dt} \xi_t = A_{u_t}^i \xi_t + \hat{w}_t, \label{line:eq}
\end{align}
where $A_{u_t}^i$ is defined by $g_{u_t}^i(\exp(\xi)) = \mathcal{L}_{\mathfrak{g}}(A_{u_t}^i \xi) + O(\norm{\xi_t}^2)$  and where we have neglected terms of order $O(\norm{\xi_t}^2)$ as well as terms of order $O(\norm{\hat w_t}\norm{\xi_t})$. The latter approximation, as well as the fact that ${\hat w_t}$ can be considered as white noise, are approximations that would require a proper justification in a stochastic setting. However,   they are part of the standard EKF methodology, see e.g.,  \cite{lefferts1982kalman} and justifying them is not the object of the present paper. Besides,   the emphasis is put here on   deterministic properties of the observer.

Regarding the output,  we consider for instance the case of left-invariant observations, and define 
$\xi_{t_n}$ through the exponential mapping \eqref{expmap}, i.e. $\exp(\xi_{t_n})=\eta_{t_n}^L$. Moreover, for $1\leq i\leq k$ let $\hat{V}_n^i$ denote  $\hat\chi_{t_n}^{-1}V_n^i$. 
The error update  \eqref{update:evolution:left:eq}, when the LIEKF update \eqref{LIEKF_update} is fed with the ``noisy" measurements  \eqref{claude:noise} becomes
\begin{align}
\left( \eta_{t_n}^L \right)^+ & = \chi_{t_n}^{-1} \hat{\chi}_{t_n}^+ \nonumber \\
 & = \eta_{t_n}^L \exp \left[ L_n
\begin{pmatrix}
\left(\eta_{t_n}^L \right)^{-1} d^1 -  d^1 +  \hat V_n^1 + \left(\eta_{t_n}^L \right)^{-1} B_n^1 \\
... \\
\left(\eta_{t_n}^L \right)^{-1} d^k -  d^k + \hat V_n^k + \left(\eta_{t_n}^L\right)^{-1} B_n^k
\end{pmatrix} \right]. \label{update:evolution:left:eq:noise}
\end{align}
To linearize it we proceed as follows. For $1\leq i\leq k$ we have
\begin{align*}
 & (\eta_{t_n})^{-1}d^i-d^i + \hat V_n^i + (\eta_{t_n}^L)^{-1} B_n^i \\
 & =\exp_m( \mathcal{L}_{\mathfrak{g}}(\xi_{t_n}))^{-1} \left( d^i+ B_n^i   \right) -d^i + \hat{V}_n^i \\
 & = \left( Id-\mathcal{L}_{\mathfrak{g}} \left( \xi \right)_{t_n} \right) \left( d^i + B_n^i \right) - d^i + \hat{V}_n^i +O\norm{{\xi}_{t_n}}^2 \\&=-\mathcal{L}_{\mathfrak{g}}(\xi)_{t_n}d^i + \hat{V}_n^i +B_n^i +O \left( \norm{{\xi}_{t_n}}^2 \right)  +O \left( \norm{{\xi}_{t_n}} \norm{{B}_{n}^i} \right),
\end{align*}
using a simple Taylor expansion of the matrix exponential map. Expanding similarly equation \eqref{update:evolution:left:eq:noise} yields 
\begin{align}
Id+\mathcal{L}_{\mathfrak{g}}(\xi_{t_n}^+) = Id+ \mathcal{L}_{\mathfrak{g}} \left( L_n
 \begin{pmatrix}
-\mathcal{L}_{\mathfrak{g}}(\xi_{t_n}) d^1 + \hat{V}_n^1 +B_n^1 \\
... \\
-\mathcal{L}_{\mathfrak{g}}(\xi_{t_n}) d^k + \hat{V}_n^k +B_n^k
\end{pmatrix} \right) + T. \label{eeeq:eq}
\end{align}with $T$ terms of order $O \left(\norm{\xi_{t_n}}^2 \right) +O \left( \norm{{\xi}_{t_n}} \norm{{B}_{n}} \right)$. 
Neglecting them  we finally get the following linearized error equation in $\RR^{\text{dim }\mathfrak{g}}$
\begin{align}
\xi_{t_n}^+ = \xi_{t_n} + L_n \left( H \xi_{t_n} + \hat{V}_n + B_n \right),\label{uppdate:eq}
\end{align}
where $H\in\RR^{kN\times\mathrm{dim}\mathfrak{g}}$, $\hat V_n\in\RR^{kN}$ and $ B_n\in\RR^{kN}$  are defined by
$$H\xi=
\begin{pmatrix}
-\mathcal{L}_{\mathfrak{g}}(\xi) d^1  \\
... \\
-\mathcal{L}_{\mathfrak{g}}(\xi) d^k  
\end{pmatrix}
,\qquad \hat V_n=
\begin{pmatrix}
\hat V_n^1 \\
... \\
\hat V_n^k
\end{pmatrix}
,\qquad B_n=
\begin{pmatrix}
 B_n^1 \\
... \\
 B_n^k
\end{pmatrix}.
$$
Now, let $\hat{Q}_t$ reflect the trusted covariance of the modified process noise $\hat w_t$, and $\hat{N}_n$ the trusted covariance of the modified measurement noise $\hat{V}_n + B_n$.  
Note that, equations \eqref{line:eq} and \eqref{uppdate:eq} mimick those of a Kalman filter designed for the following auxiliary linear  system with discrete measurements: $\dotex x_t=A_{u_t} x_t + \hat{w}_t$, $y_n=H x_{t_n}+\hat{V}_n + B_n$. The standard Kalman theory thus suggests to compute $L_n$ through the Riccati equation 
%\begin{equation}
%\label{Kgains:eq}
%\begin{aligned}
%\frac{d}{dt} P_t  = F_{u_t} P_t + P_t F_{u_t}^T +\hat{Q}_t \qquad , \qquad
%S_n  = H P_{t_n} H^T + \hat{N}_n \\
%L_n  = P_{t_n} H^T S^{-1} \qquad , \qquad
%P_{t_n}^+  = (I-L_n H) P_{t_n}
%\end{aligned}
%\end{equation}
\begin{equation}
\label{Kgains:eq}\begin{gathered}
\frac{d}{dt} P_t  = A_{u_t} P_t + P_t A_{u_t}^T +\hat{Q}_t , \qquad P_{t_n}^+  = (I-L_n H) P_{t_n},
\\ \text{with }  
S_n  = H P_{t_n} H^T + \hat{N}_n ,~
L_n  = P_{t_n} H^T S^{-1}.\end{gathered}
\end{equation}

\subsection{Summary of IEFK equations}
In a deterministic context, the IEKF equations can be compactly recapped as follows
\begin{equation}
\begin{aligned}
&\frac{d}{dt}\hat{\chi}_t  = f_{u_t}(\hat{\chi}_t),\quad  t_{n-1}\leq t< t_n , \\ 
&\hat{\chi}_{t_n}^+  = \hat{\chi}_{t_n} \exp \left[ L_n
\begin{pmatrix}
\hat{\chi}_{t_n}^{-1} Y_{t_n}^1- d^1 \\
... \\
\hat{\chi}_{t_n}^{-1} Y_{t_n}^k - d^k
\end{pmatrix} \right], ~\text{(LIEKF)},\\ \text{\underline{or}}\quad & \hat{\chi}_{t_n}^+  = \exp \left[ L_n
\begin{pmatrix}
\hat{\chi}_{t_n} Y_{t_n}^1- d^1 \\
... \\
\hat{\chi}_{t_n} Y_{t_n}^k - d^k
\end{pmatrix} \right] \hat{\chi}_{t_n},  ~\text{(RIEKF)}
\end{aligned} \label{LIEKF:::eq}\end{equation}
where the LIEKF (resp. RIEKF) is to be used in the case of left (resp. right) invariant outputs. The gain $L_n$ is obtained in each case through the following Riccati equation
\begin{equation}
\label{eq::Riccati_Left}
\begin{gathered}
\frac{d}{dt} P_t  = A_{u_t} P_t + P_t A_{u_t}^T +\hat{Q}_t , \\
S_n  = H P_{t_n} H^T + \hat{N}_n , \quad
L_n  = P_{t_n} H^T S^{-1} , \quad
P_{t_n}^+ = (I-L_n H) P_{t_n}.
\end{gathered}
\end{equation}

As concerns the LIEKF, $A_{u_t}$ is defined by $g^L_{u_t}(\exp(\xi)) = \mathcal L_{\mathfrak{g}}(A_{u_t} \xi) + O (\norm{\xi}^2)$, and $H\in\RR^{kN\times\mathrm{dim}\mathfrak{g}}$ is defined by $H\xi=
\begin{pmatrix}
-\left( \mathcal{L}_{\mathfrak{g}} (\xi) d^1 \right)^T,
...,
-\left( \mathcal{L}_{\mathfrak{g}} (\xi) d^k  \right)^T
\end{pmatrix}^T$.  
The design matrix parameters $\hat Q_t, \hat{N}_n,\hat{V}_n, B_n$ are freely assigned by the user. When sensor noise characteristics are known, they can provide the user with a first sensible tuning of those matrices by considering the associated ``noisy" system \eqref{eq::IEKF_model}-\eqref{claude:noise} of Section \ref{ouida}. In this case, the matrices can be \emph{interpreted} in the following way:   $\hat{Q}_t\in\RR^{\text{dim }\mathfrak{g}\times \text{dim }\mathfrak{g}}$  denotes the covariance  of the modified process noise $\hat w_t=-\mathcal L^{-1}_{\mathfrak{g}}(w_t)$ and $\hat{N}_n$ the covariance matrix of the noise $\hat{V}_n+ B_n$, $\hat{V}_n$ and $B_n$ being defined as:
 $ \hat V_n= ( \hat{\chi}_{t_n}^{-1}V_n^1, ... ,
\hat{\chi}_{t_n}^{-1}V_n^k )^T$ , $B_n= (
 B_n^1 , ... ,  B_n^k )^T$.

As concerns the RIEKF implementation, $A_{u_t}$ is defined by $g^R_{u_t}(\exp(\xi)) = \mathcal L_{\mathfrak{g}}(A_{u_t} \xi) + O (\norm{\xi}^2)$, and $H\in\RR^{kN\times\mathrm{dim}\mathfrak{g}}$ by $H\xi=
\begin{pmatrix}
\left( \mathcal{L}_{\mathfrak{g}} (\xi) d^1 \right)^T,
... ,
\left( \mathcal{L}_{\mathfrak{g}} (\xi) d^k  \right)^T
\end{pmatrix}^T$. The design matrix parameters $\hat Q_t, \hat{N}_n,\hat{V}_n, B_n$ are freely assigned by the user. Considering the associated ``noisy" system \eqref{eq::IEKF_model}-\eqref{claudio:noise} of Section \ref{ouida} they can be interpreted as follows. $\hat{Q}_t$ denotes the covariance of the modified process noise  $\hat{w}_t = - Ad_{\hat{\chi}_t} \mathcal{L}_{\mathfrak{g}}^{-1} (w_t)$  and $\hat{N}_n$ the covariance matrix of the noise ${V}_n+ \hat{B}_n$, ${V}_n$ and $\hat B_n$ being defined as:
 $ V_n = \left(  V_n^1, ... ,
 V_n^k \right)^T$ , $\hat B_n = \left(
 \hat{\chi}_{t_n}B_n^1 , ... ,  \hat{\chi}_{t_n}B_n^k \right)^T$.

\subsection{Stability properties}
The aim of the present section is to study the stability properties of the IEKF as a \emph{deterministic} observer for the system \eqref{eq::IEKF_modelle}-\eqref{claude:eq}, or alternatively \eqref{eq::IEKF_modelle}-\eqref{claudio:eq}.  We are about to prove the IEKF has  \emph{guaranteed} (local) stability properties, that  rely on the error equation properties   the  EKF does \emph{not} possess. The stability of an observer is defined as its ability to recover from a perturbation or an erroneous initialization:
\begin{defn}\label{deff}
Let $(x_0,t_0,t) \rightarrow X_{t_0}^t(x_0)$ denote a continuous flow on a space $\mathcal{X}$ endowed with a distance $d$. A flow $(z,t_0,t) \rightarrow \hat{X}_{t_0}^t(z)$ is an asymptotically stable observer of $X$ about the trajectory $\left( {X}_{t_0}^t ({x}) \right)_{t\geq t_0}$  if there exists $\epsilon>0$ such that:
\[
d \left( x, \tilde{x} \right) <\epsilon \Rightarrow   d \left( X_{t_0}^t (x),\hat{X}_{t_0}^t (\tilde{x}) \right) \rightarrow 0 \text{ when }t\to + \infty.
\]
\end{defn}
Theorem \ref{thm::localCV} below is the main result of the paper . It is a consequence of Theorem \ref{dechire:thm} of Section \ref{sect::deterministic}.  
J. J. Deyst and C. F. Price have shown in \cite{deyst} the following theorem, stating sufficient conditions for the Kalman filter to be a stable observer for \emph{linear} (time-varying) deterministic systems.

\begin{thm}[Deyst and Price, 1968]
\label{thm::Deyst_Price}Consider the  linear system $\dotex x_t=A_t x_t$, $y_{t_n}=Hx_{t_n}$ with $x_t \in \RR^p$ and let $\Phi_{{t_0}}^{t}$ denote the square matrix defined by $\Phi_{t_0}^{t_0}=I_p, \frac{d}{dt} \Phi_{t_0}^t = A_t \Phi_{t_0}^t$.
If there exist $\alpha_1,\alpha_2,\beta_1,\beta_2,\delta_1,\delta_2,\delta_3,M$ such that: 
\begin{enumerate}[i]
\item $(\Phi_{t_n}^{t_{n+1}})^T \Phi_{t_n}^{t_{n+1}} \succeq \delta_1 I_p \succeq 0,$ \label{cond::Psi}
\item $\exists q \in \NN^{*},\forall s >0, \exists G_s \in \RR^{p \times q}, Q_s=G_s Q'G_s^T$ where $Q' \succeq \delta_2 I_q \succeq 0,$ \label{cond::Q}
\item $N_n \succeq \delta_3 I \succeq 0,$ \label{cond::R}
\item $\alpha_1 I_p \leq \int_{s=t_{n-M}}^{t_n} \left( \Phi_s^{t_n} \right) Q_s \left( \Phi_s^{t_n} \right)^T \leq \alpha_2 I_p,$ \label{cond::reachable}
\item $\beta_1 I_p \leq \sum_{i=n-M}^{n-1} \left( \Phi^{t_n}_{t_{i+1}} \right)^T H^TN_n^{-1}H \left( \Phi^{t_n}_{t_{i+1}} \right) \leq \beta_2 I_p.$ \label{cond::observable}
\end{enumerate}
Then the linear Kalman filter tuned with covariance matrices $Q$ and $N$ is an asymptotically stable observer for the Euclidean distance. More precisely there exist $\gamma_{\text{min}},\gamma_{\text{max}}>0$ such that $\gamma_{\text{min}} I\preceq P_t I\preceq\gamma_{\text{max}}I$ for all $t$ and $(\hat x_t-x_t)^TP_t^{-1}(\hat x_t-x_t)$ has  exponential decay.
\end{thm}

The main theorem of the present paper is the extension of this linear result to the non-linear case when the Invariant Extended Kalman Filter is used for systems of Section \ref{sect::deterministic}.
\begin{thm}
\label{thm::localCV}
Consider the system \eqref{eq::IEKF_modelle}-\eqref{claude:eq} (respectively  \eqref{eq::IEKF_modelle}-\eqref{claudio:eq}). Suppose the stability conditions of the linear Kalman filter given in Theorem \ref{thm::Deyst_Price} are verified  about the \emph{true} system's trajectory $\chi_t$ (i.e. are verified for the linear system obtained by linearizing the system \eqref{eq::IEKF_modelle} with left (resp. right) invariant output about $\chi_t$). Then the Left (resp. Right) Invariant Extended Kalman Filter  estimate $\hat{\chi}_t$ defined at eq. \eqref{LIEKF:::eq} is an asymptotically stable observer of $\chi_t$ in the sense of Definition \ref{deff}. Moreover, the convergence radius $\epsilon>0$ is valid over the whole trajectory (i.e. is independent of the initialization time $t_0$).
\end{thm}
\begin{proof}
The full proof is technical and has been moved to Appendix \ref{proof::localCV}.  The rationale is to compare the evolution of the logarithmic error $\xi_t$ defined as $\eta_t = \exp(\xi_t)$, with its linearization. For the general EKF, the control of second and higher order terms in the error equation is difficult because: 1- they depend on the inputs $u_t$ 2- they depend on the linearization point $\hat{\chi}_t$ 3- the estimation error impacts the gain matrices. 
For the IEKF, the main difficulties vanish as during the propagation step the IEKF is built for the  logarithmic error $\xi_t$ whose evolution $\frac{d}{dt} \xi_t  = A_{u_t} \xi_t$ is in fact \emph{exact}  (no higher order terms) due to  Theorem \ref{dechire:thm}. At the update step, due to the specific form of the IEKF update, second order terms can be \emph{uniformly}  bounded over $n$. And finally, due to the error equation of the IEKF, the Riccati equation depends on the estimate \emph{only} through the  matrices $\hat{Q}_t$ and $\hat{N}_t$ which affect stability in a minor way, as shown by  Theorem \ref{thm::Deyst_Price}.

%In other words the error variable $\eta_t = \exp(\xi_t)$ is governed by:
%\begin{align}
%\frac{d}{dt} \xi_t  = A_{u_t} \xi_t \label{eq::linearized_dynamic},\qquad 
%\xi_{t_n}^+  = (I-L_nH)\xi_{t_n} + r_n(\xi_{t_n}) 
%\end{align}
%where $\Psi_{t_0}^t$ denote the flow of the linear part of the equations governing $\xi_t$ (that is, $
%\Psi_{t_0}^{t_0} = Id , \quad 
%\frac{d}{dt} \Psi_{t_0}^t = A_{u_t} \Psi_{t_0}^t ,\quad
%\Psi_{t_0}^{t+} = (I-L_nH)\Psi_{t_0}^t$) and $r_{t_n}(\xi)$ is a second-order rest.
%The solution of \eqref{eq::linearized_dynamic} can be decomposed as follows:
%\begin{equation}
%\forall t\geq 0,\qquad\xi_t = \Psi_0^t \xi_0 + \sum_{t_n<t}\Psi_{t_n}^t r_n(\xi_{t_n}) \label{eq::decomposition}
%\end{equation}
%For any initialization $\xi_0$ the flow $\Psi_{0}^t$ has an exponential decay  (Theorem \ref{thm::Deyst_Price}). All we have to ensure is that the apparition of the second-order terms $r_{t_n}(\xi_{t_n})$ is compensated by this decay with the additional difficulty that noise covariance matrices depend on the \emph{estimated} trajectory. 
\end{proof}
%The full proof is technical and has been removed to Appendix \ref{proof::localCV}.  The rationale is as follows: the error variable $\eta_t = \exp(\xi_t)$ is governed by  \eqref{line:eq}-\eqref{uppdate:eq}, that is
%\begin{align}
%\frac{d}{dt} \xi_t  = A_{u_t} \xi_t \label{eq::linearized_dynamic},\qquad 
%\xi_{t_n}^+  = (I-L_nH)\xi_{t_n} + r_n(\xi_{t_n}) 
%\end{align}
%where $r_{t_n}(\xi)$ is second-order. Let $\Psi_{t_0}^t$ denote the flow of the linear part of the equations, that is, $
%\Psi_{t_0}^{t_0} = Id , \quad 
%\frac{d}{dt} \Psi_{t_0}^t = A_{u_t} \Psi_{t_0}^t ,\quad
%\Psi_{t_0}^{t+} = (I-L_nH)\Psi_{t_0}^t$. 
%The solution of \eqref{eq::linearized_dynamic} can be decomposed as follows:
%\begin{equation}
%\forall t\geq 0,\qquad\xi_t = \Psi_0^t \xi_0 + \sum_{t_n<t}\Psi_{t_n}^t r_n(\xi_{t_n}) \label{eq::decomposition}
%\end{equation}
%For any initialization $\xi_0$ the flow $\Psi_{0}^t$ has an exponential decay  (Theorem \ref{thm::Deyst_Price}). All we have to ensure is that the apparition of the second-order terms $r_{t_n}(\xi_{t_n})$ is compensated by this decay with the additional difficulty that noise covariance matrices depend on the \emph{estimated} trajectory. 

The result displayed in Theorem \ref{thm::localCV} is in sharp contrast with the usual results available, that make the \emph{highly non-trivial  assumption} that the linearized system around the \emph{estimated} trajectory is well-behaved \cite{boutayeb,song-grizzle-95,reif,bonnabel2012contraction}. But this fact is almost impossible to predict as when the estimate is (even slightly) away  from the true state, the Kalman gain becomes erroneous, which can in turn amplify the discrepancy between estimate and true state so that there is \emph{no} reason the assumption should keep holding. On the other hand, when considering an actual system undergoing a realistic physical motion (see the examples below), if sufficiently many sensors are available, one can generally assert \emph{in advance} the linearized system around the \emph{true} trajectory  possesses all the desired properties. The following consequence  proves useful in practice.
%\begin{thm}
%\label{thm::bounded}
%If \ref{cond::Psi} is verified, and the system linearized about the \emph{true} trajectory is uniformly  observable and the eigenvalues of the covariance matrices $\hat Q_t,\hat N_n$ are uniformly lower and upper bounded by strictly positive scalars, the conditions of Theorem \ref{thm::localCV} are satisfied and thus the IEKF is asymptotically stable.
%\end{thm}
\begin{thm}
\label{thm::bounded}
Assume the system linearized around the \emph{true} trajectory has the following properties : the propagation matrix $A_t$ is constant, there exist matrices $B,D$ such that $\hat{Q}=B \bar{Q} B^T$ and $\hat{N} = D \bar{N} D^T$ with $\bar{Q}$ and $\bar{N}$ upper- and lower-bounded, with $(A,H,B,D)$ detectable and reachable. Then the conditions of Theorem \ref{thm::localCV} are satisfied and the IEKF is asymptotically stable.
\end{thm}

%%%%%%%%%%%%%%%%%%%%%%%%%%%%%%%%%%%%%%%%%%%%%%%%%%%%%
%  CAR
%%%%%%%%%%%%%%%%%%%%%%%%%%%%%%%%%%%%%%%%%%%%%%%%%%%%%

\section{Simplified car example}\label{sect::examples:A}
The computations require only  basic knowledge about matrix Lie groups as recalled in Appendix \ref{sect::tuto_Lie_groups}.

\subsection{Considered model}

%\begin{figure}[h]
%\centering
%\includegraphics[width=11cm]{Car_SE2}
%\end{figure}

Consider a (non-holonomic) car evolving on the 2D plane. Its heading is denoted by an angle $\theta_t \in [-\pi,\pi]$ and its position by a vector $x_t \in \RR^2$. They follow the classical equations (see, e.g., \cite{de1998feedback}):
\begin{align}\label{car:dyna}
\frac{d}{dt} \theta_t  = u_t v_t , \qquad
\frac{d}{dt} x^1_t  = \cos(\theta_t) v_t , \qquad
\frac{d}{dt} x^2_t  = \sin(\theta_t) v_t ,
\end{align}
where $v_t$ is the velocity measured by an odometer and $u_t$ (a function of) the steering angle. Two kinds of observations are considered:
\begin{align}
\bar{Y}_n & = x_{t_n} \label{car_GPS:nonoise} \\
\text{or } \bar{Y}_n^k & = R(\theta_{t_n})^T(x_{t_n}-p_k), \qquad k \in [1, K]\label{car_feature:nonoise}
\end{align}
where $R(\theta)$ is a planar rotation of angle $\theta$. Equation
\eqref{car_GPS:nonoise} represents a  position measurement (GPS for instance) whereas \eqref{car_feature:nonoise} represents a range-and-bearing observation of a sequence of known features located at $p_k \in \RR^2$ for $k \in [1, K]$.

\subsection{IEKF gain tuning}
\label{sect::gain_tuning_car}

To derive and tune the IEKF equations, we follow the methodology of Section \ref{ouida} which amounts to 1- associate a ``noisy "system to the original considered system, just because it allows obtaining a sensible  tuning of the design matrices from an engineering viewpoint, 2- transform it into a system defined on a matrix Lie group to make it fit into our framework, 3- linearize the ``noisy" equations, and 4- use the Kalman equations to tune the observer gain. 

\subsubsection{Associated ``noisy" system}
Taking into account the possible noise in the measurements we get
\begin{equation}\begin{gathered}\label{mec:eq}
\frac{d}{dt} \theta_t  = u_t v_t +w_t^{\theta}, \qquad
\frac{d}{dt} x^1_t  = \cos(\theta_t) (v_t + {w^l_t}) - \sin(\theta_t) {w^{tr}_t}, \\
\frac{d}{dt} x^2_t  = \sin(\theta_t) (v_t + {w^l_t}) + \cos(\theta_t){w^{tr}_t},
\end{gathered}\end{equation}
with $w_t^\theta$ the differential odometry error, ${w_t^l}$ the longitudinal odometry error and $w_t^{tr}$ the transversal shift. By letting noise enter the measurement equations we get the following two kinds of measurements:
\begin{align}
\bar{Y}_n & = x_{t_n} + V_n \label{car_GPS} \\
\text{or } \bar{Y}_n^k & = R(\theta_{t_n})^T(x_{t_n}-p_k) + \bar{V}_n^k, \qquad k \in [1, K].\label{car_feature}
\end{align}

\subsubsection{Matrix form}
This system can be embedded in the matrix Lie group $SE(2)$ (see Appendix \ref{sect::tuto_SE2}) using the   matrices:
$$
{ \chi_t = \begin{pmatrix} \cos(\theta_t) & -sin(\theta_t) & x^1_t \\ \sin(\theta_t) & \cos(\theta_t) & x^2_t \\ 0 & 0 & 1 \end{pmatrix},}$$
$${\nu_t = \begin{pmatrix} 0 & -u_t v_t & v_t \\ u_t v_t & 0 & 0 \\ 0 & 0 & 0 \end{pmatrix},} \quad {w_t = \begin{pmatrix} 0 &  -w^\theta_t & w_t^l \\ w_t^\theta & 0 & {w_t^{tr}} \\
0 & 0 & 0 \end{pmatrix} }.
$$
 The equation \eqref{mec:eq} governing the ``noisy" system evolution writes:
\begin{equation}
\frac{d}{dt} \chi_t = \chi_t (\nu_t + w_t), \label{eq::car_matrix}
\end{equation}
and the observations \eqref{car_GPS} and \eqref{car_feature}  respectively have the equivalent form:
\begin{equation}
\label{eq::car_matrix_GPS}
Y_n = \begin{pmatrix} x_{t_n} + V_n \\ 1 \end{pmatrix} = \chi_{t_n} \begin{pmatrix} 0_{2 \times 1} \\ 1 \end{pmatrix} + \begin{pmatrix} V_n \\ 0 \end{pmatrix},
\end{equation}

\begin{equation}
\label{eq::car_matrix_feature}
Y_n^k = \begin{pmatrix} R(\hat{\theta}_{t_n})^T(x_{t_n}-p_k) \\ 1 \end{pmatrix} + \begin{pmatrix} \bar{V}_n^k \\ 0 \end{pmatrix} = - \chi_{t_n}^{-1} \begin{pmatrix} p_k \\ 1 \end{pmatrix} + \begin{pmatrix} \bar V_n^k \\ 0 \end{pmatrix}.
\end{equation}

The reader can verify relation \eqref{eq::main_relation} letting $f_{\nu_t}(\chi_t)=\chi_t\nu_t$.

%%%%%%%%%%%%%%%%%%%%%%%%%%%%%%

\subsubsection{IEKF equations for the left-invariant output \eqref{car_GPS:nonoise}}
\label{ex::car_GPS}
The LIEKF  equations \eqref{LIEKF:::eq} for the associated ``noisy" system \eqref{eq::car_matrix}, \eqref{eq::car_matrix_GPS} write:
\begin{equation*}
\frac{d}{dt} \hat{\chi}_t  = \hat{\chi}_t \nu_t 
 , \qquad
\hat{\chi}^+_{t_n}  = \hat{\chi}_{t_n} \exp \left( L_n \left[ \hat{\chi}_{t_n}^{-1} Y_n -
\begin{pmatrix}
0_{2 \times 1} \\ 1
\end{pmatrix}
\right] \right).
\end{equation*}
As the bottom element of $ \left[ \hat{\chi}_{t_n}^{-1} Y_n -
\begin{pmatrix}
0_{2 \times 1} \\ 1
\end{pmatrix} \right]$ is always zero we can conveniently use a reduced-dimension gain matrix $\tilde{L}_n$ defined by $L_n = \tilde{L}_n \tilde{p}$ with $\tilde{p}=(I_2, 0_{2,1})$. To compute the gains, we write the left-invariant error  $
\eta_t = \chi_t^{-1} \hat{\chi}_t
$
whose evolution is:
\begin{equation}\begin{aligned}
&\frac{d}{dt} \eta_t = \eta_t \nu_t - \nu_t \eta_t - w_t \eta_t,\\&~~\eta_{t_n}^+ = \eta_{t_n} \exp \left( \tilde{L}_n \tilde{p} \left[ \eta_{t_n}^{-1} \begin{pmatrix} 0_{2 \times 1} \\ 1 \end{pmatrix}- \begin{pmatrix} 0_{2 \times 1} \\ 1 \end{pmatrix}+\hat{\chi}_{t_n}^{-1} \begin{pmatrix} V_n \\ 0 \end{pmatrix} \right] \right).\label{expl::car}
\end{aligned}\end{equation}
To linearize this equation we introduce the linearized error $\xi_t$ defined replacing $\eta_t$ with $I_3 + \mathcal{L}_{\mathfrak{se}(2)}(\xi_t)$. Introducing the first-order approximations $\eta_t = I_3 + \mathcal{L}_{\mathfrak{se}(2)}(\xi_t)$, $\eta_t^+ = I_3 + \mathcal{L}_{\mathfrak{se}(2)}(\xi_t^+)$, $\exp(u)=I_3+\mathcal{L}_{\mathfrak{se}(2)}(u)$ and $\eta_t^{-1} = I_3 - \mathcal{L}_{\mathfrak{se}(2)}(\xi_t) $ in \eqref{expl::car} and removing the second-order terms in $\xi_t$, $V_n$ and $w_t$ we obtain:
$$
\frac{d}{dt} \xi_t  =  -\begin{pmatrix} 0 & 0 & 0 \\
0 & 0 & -u_t v_t \\
-v_t & u_t v_t & 0 \end{pmatrix} \xi_t - \begin{pmatrix}
w_t^\theta \\
w_t^l \\
w_t^{tr} 
\end{pmatrix}$$ $$ \text{and} \qquad
\xi_{t_n}^+  = \xi_{t_n} - \tilde{L}_n \left[ (0_{2,1}, I_2) \xi_t - R \left( \hat \theta_{t_n} \right)^T V_n \right].$$ The gains $\tilde{L}_n$ are thus finally computed using the Riccati equation \eqref{eq::Riccati_Left} with:
$$ A_t = -\begin{pmatrix} 0 & 0 & 0 \\
0 & 0 & -u_t v_t \\
-v_t & u_t v_t & 0 \end{pmatrix}, \quad H = (0_{2,1}, I_2),$$ $$
\hat{Q}_t = Cov[ (
w_t^\theta,
w_t^l ,
w_t^{tr})^T], \quad
\hat{N} = R \left( \hat{\theta}_{t_n} \right) Cov \left( V_n \right) R \left( \hat{\theta}_{t_n} \right)^T.
$$

%%%%%%%%%%%%%%%%%%%%%%%%%%%%%%%%%%%%%%%%%%%%%%%%%%%%%%%

\subsubsection{IEKF equations for the right-invariant output \eqref{car_feature}}\label{etouais}

The RIEKF equations \eqref{LIEKF:::eq} for the associated ``noisy" system \eqref{eq::car_matrix}, \eqref{eq::car_matrix_feature} write:
\begin{equation*}
\frac{d}{dt} \hat{\chi}_t  = \hat{\chi}_t \nu_t
~,~ 
\hat{\chi}_{t_n}^+  = \exp \left( L_n \left[ \hat{\chi}_{t_n} Y_n^1 +
\begin{pmatrix}
p_1 \\
1
\end{pmatrix}
; ... ; \hat{\chi}_{t_n} Y_n^K +
\begin{pmatrix}
p_K \\
1
\end{pmatrix}
\right] \right) \hat{\chi}_{t_n}.
\end{equation*}
As the bottom element of $ \left[ \hat{\chi}_{t_n}^{-1} Y_n^k +
\begin{pmatrix}
p_k \\ 1
\end{pmatrix} \right]$ is always zero we can conveniently use a reduced-dimension gain matrix $\tilde{L}_n$ defined by $L_n = \tilde{L}_n \tilde{p}$ with $\tilde{p}= \begin{pmatrix}[I_2, 0_{2,1}]  & & \\ & \ddots & \\ & & [I_2, 0_{2,1}] \end{pmatrix}$. To compute the gains $\tilde{L}_n$ we derive the evolution of the right-invariant error variable 
$
\eta_t = \hat{\chi}_t \chi_t^{-1}
$ between the estimate and the state of the associated ``noisy" system:
\begin{equation}
\label{eq::erro_car_feature}
\begin{aligned}
\frac{d}{dt} \eta_t  & = - (\hat{\chi}_t w_t \hat{\chi}_t^{-1}) \eta_t, \\
\eta_{t_n}^+ & = \exp \left( \tilde{L}_n \tilde{p} \begin{pmatrix} -\eta_{t_n} \begin{pmatrix} p_1 \\ 1 \end{pmatrix} + \begin{pmatrix} p_1 \\ 1 \end{pmatrix} + \hat{\chi}_{t_n} \begin{pmatrix} V_n^1 \\ 0 \end{pmatrix} \\ ... \\  -\eta_{t_n} \begin{pmatrix} p_K \\ 1 \end{pmatrix}+ \begin{pmatrix} p_K \\ 1 \end{pmatrix} +   \hat{\chi}_{t_n} \begin{pmatrix}  V_n^K\\ 0 \end{pmatrix} \end{pmatrix} \right) \eta_{t_n} .
\end{aligned}
\end{equation}
To linearize this equation we introduce the linearized error $\xi_t$ defined as $\eta_t = I_3 + \mathcal{L}_{\mathfrak{se}(2)}(\xi_t)$. Introducing $\eta_t = I_3 + \mathcal{L}_{\mathfrak{se}(2)}(\xi_t)$, $\eta_t^+ = I_3 + \mathcal{L}_{\mathfrak{se}(2)}(\xi_t^+)$, $\exp(u)=I_3+\mathcal{L}_{\mathfrak{se}(2)}(u)$ and $\eta_t^{-1} = I_3 - \mathcal{L}_{\mathfrak{se}(2)}(\xi_t) $ in \eqref{eq::erro_car_feature} and removing the second-order terms in $\xi_t$, $V_n$ and $w_t$ we obtain:
\begin{align*}
\frac{d}{dt} \xi_t &  = - 
\begin{pmatrix}
1 & 0_{1,2} \\
\begin{matrix}
\hat{x}_t^2 \\
-\hat{x}_t^1
\end{matrix} & R(\hat{\theta_t})
\end{pmatrix} \begin{pmatrix} w_t^\theta \\ w_t^l \\ w_t^{tr} \end{pmatrix},\\
\xi_{t_n}^+ & = \xi_{t_n} - \tilde{L}_n \left[ \begin{pmatrix}
\begin{pmatrix}
-(p_1)_2 & 1 & 0 \\
 (p_1)_1 & 0 & 1
\end{pmatrix}  \\
 ... \\
\begin{pmatrix}
-(p_K)_2 & 1 & 0 \\
 (p_K)_1 & 0 & 1
\end{pmatrix}  
 \end{pmatrix} \xi_{t_n} -
 \begin{pmatrix}
 R(\hat{\theta}_{t_n})  V_n^1 \\
 ... \\
 R(\hat{\theta}_{t_n}) V_n^K
 \end{pmatrix} \right].
\end{align*}
 The gains are thus computed using the Riccati equation \eqref{eq::Riccati_Left} with $A_t,H, \hat{Q}$ and $\hat{N}$ defined as:
\[
A_t= 0_{3,3}, \quad H_n = \begin{pmatrix}
\begin{pmatrix}
-(p_1)_2 & 1 & 0 \\
 (p_1)_1 & 0 & 1
\end{pmatrix}  \\
 ... \\
\begin{pmatrix}
-(p_K)_2 & 1 & 0 \\
 (p_K)_1 & 0 & 1
\end{pmatrix}  
 \end{pmatrix},\]
 \[
\hat{Q}_t = \begin{pmatrix}
1 & 0_{1,2} \\
\begin{matrix}
\hat{x}_t^2 \\
-\hat{x}_t^1
\end{matrix} & R(\hat{\theta}_t) 
\end{pmatrix} Cov \begin{pmatrix} w_t^\theta \\ w_t^l \\ w_t^{tr} \end{pmatrix} \begin{pmatrix}
1 & 0_{1,2} \\
\begin{matrix}
\hat{x}_t^2 \\
-\hat{x}_t^1
\end{matrix} & R(\hat{\theta}_t) 
\end{pmatrix}^T,
\]
\[
\hat{N}_n = \begin{pmatrix} R(\hat{\theta}_{t_n}) Cov (N^1) R(\hat{\theta}_{t_n})^T  &  & 0 \\  & \ddots &  \\
0 &  & R(\hat{\theta}_{t_n}) Cov ( N^K ) R(\hat{\theta}_{t_n})^T \end{pmatrix}.
\]

 \subsection{Stability properties of the IEKF viewed as a non-linear observer for the simplified car}
 
 \subsubsection{Stability of the IEKF for the left-invariant output \eqref{car_GPS:nonoise}}

\begin{prop}
\label{prop::car}
If  there exists $v_{\max},v_{\min}>0$ such that the displacement satisfies $\norm{x_{t_{n+1}}-x_{t_n}}\geq v_{\min}>0$ and the input velocity satisfies  $u_t \leq v_{\max}$ then the LIEKF derived at Section \ref{etouais} is an asymptotically stable observer in the sense of Definition \ref{deff} about \emph{any} trajectory.
\end{prop}

The proof is a verification of the hypotheses of Theorem \ref{thm::localCV} and  has been moved to Appendix \ref{proof::car}.  
Note that it seems very difficult to improve on the assumptions: if the car is at the same place each time its position is measured, the heading $\theta_t$ becomes unobservable, and in practice an arbitrary high velocity is unfeasible. 
 \subsubsection{Stability of the IEKF for the right-invariant output \eqref{car_feature:nonoise}}

%\begin{prop}
%If the true trajectory is bounded and if at least two distinct points are observed then the IEKF is a stable observer.
%\end{prop}

\begin{prop}
If at least two distinct points are observed then the IEKF is an asymptotically stable observer in the sense of Definition \ref{deff} about \emph{any} bounded trajectory.
\end{prop}

\begin{proof}
According to Theorem \ref{thm::bounded} it is sufficient to show that in this case the observation matrix $H$ is full-rank, i.e. of rank 3. This is obvious as the position and the heading are easily computed from the observation of two vectors at known locations.
\end{proof}

\subsection{Simulations}

The IEKF described in Section \ref{sect::gain_tuning_car} has been implemented and compared to a classical EKF for the experimental setting described by Figure \ref{fig::Simu_2D}. The car drives along a 10-meter diameter circle  for 40 seconds with high rate odometer measurements (100 Hz) and low rate GPS measurements (1 Hz). 

The equations of the IEKF can be found above. The conventional  EKF  is based on the linear error $e_t=\begin{pmatrix} \theta_t - \hat \theta_t \\  x_t^1 - \hat x_t^1 \\ x_t^2 -\hat x_t^2 \end{pmatrix}$ yielding the linearized matrices  $
F_t = \begin{pmatrix}
0 & 0 & 0 \\
-\sin(\hat \theta_t) v_t & 0 & 0 \\
\cos(\hat \theta_t) v_t & 0 & 0
\end{pmatrix} $ and $
H_n = 
\begin{pmatrix}
0 & 1 & 0 \\
0 & 0 & 1 
\end{pmatrix}.
$
 Both filters are tuned with the same design parameters (which can be interpreted as odometer and GPS noise covariances) $
N = I_2,$ and $Q=diag( (\pi/180)^2,10^{-4} ,10^{-4})$ i.e. moderate angular velocity uncertainty and highly precise linear velocity. 
The simulation is performed for two initial values of the heading error: 1$^{\circ}$ and 45$^{\circ}$ while the initial position is always assumed known. The covariance matrix $P_0$ is consistent with the initial error (it encodes a standard deviation of the heading of 1$^{\circ}$ and 45$^{\circ}$ respectively).

The results are displayed on Figure \ref{fig::Simu_2D}. We see that for small initial errors both filters behave similarly for a long time, but for larger errors they soon behave differently, and we see  the IEKF, whose design has been adapted to the specific structure of the system, completely outperforms the EKF.

\begin{figure*}[!h]
  \centering
\begin{minipage}[b][]{0.49\linewidth}\noindent
\includegraphics[width=\linewidth]{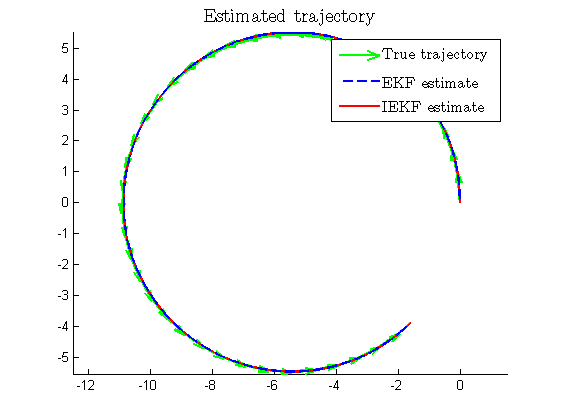}
\end{minipage}
\begin{minipage}[b][]{0.49\linewidth}\noindent
\includegraphics[width=\linewidth]{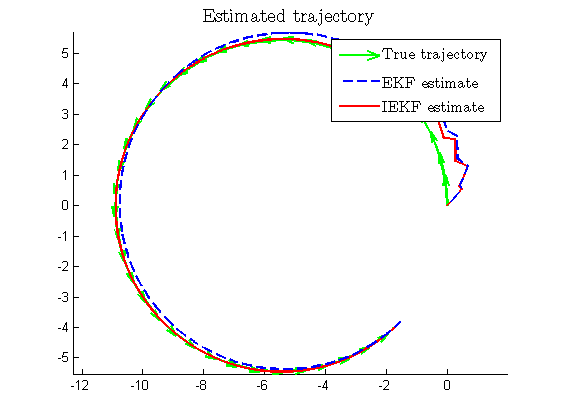}
\end{minipage}
\begin{minipage}[b][]{0.49\linewidth}
\includegraphics[width=\linewidth]{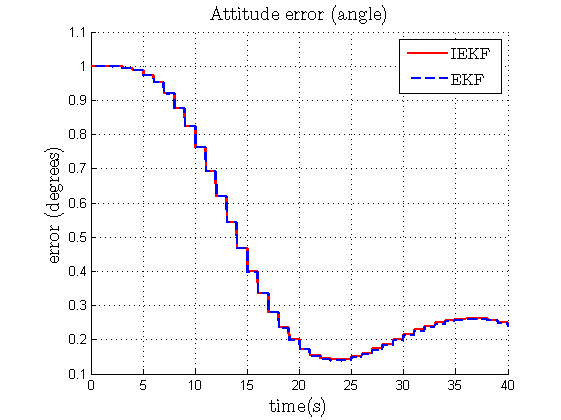}
\end{minipage}
\begin{minipage}[b][]{0.49\linewidth}
\includegraphics[width=\linewidth]{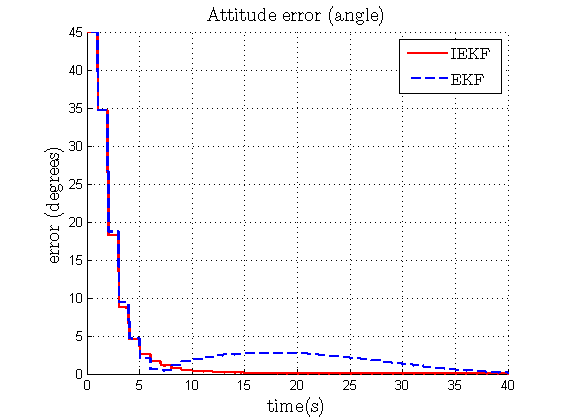}
\end{minipage}
\begin{minipage}[b][]{0.49\linewidth}\noindent
\includegraphics[width=\linewidth]{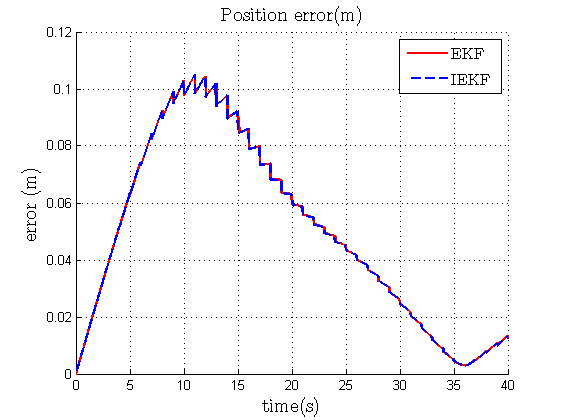}
\end{minipage}
\begin{minipage}[b][]{0.49\linewidth}\noindent
\includegraphics[width=\linewidth]{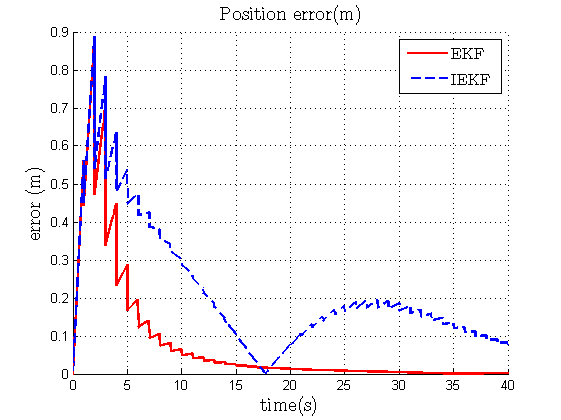}
\end{minipage}
  \caption{The heading and position of the car are estimated through EKF and IEKF with high rate odometry and low rate GPS measurements. Top plots illustrate the experimental setting and display the estimated trajectories, middle plots display the heading errors and bottom plots the position errors. As the starting point is assumed known in this simulation, the initial values of the latter are zero. But it increases afterwards due to initial heading error.
\textbf{Left column:} small initial angle error (1$^{\circ}$). We see EKF and IEKF behave similarly (at least for a long time) as propagation steps are identical. \textbf{Right column:} large initial angle error (45$^{\circ}$). The behaviors rapidly become different, and the EKF is  outperformed. Due to its righteous use of the system's non-linearities, the IEKF keeps ensuring rapid estimation error decrease.
}
\label{fig::Simu_2D}
\end{figure*}

%%%%%%%%%%%%%%%%%%%%%%%%%%%%%%%%%%%%%%%%%%%%%%%%%%%%%
%  DOUBLE SE(3)
%%%%%%%%%%%%%%%%%%%%%%%%%%%%%%%%%%%%%%%%%%%%%%%%%%%%%

\section{Navigation on flat earth}\label{sect::examples:B}

In this example we estimate the orientation, velocity, and position of a rigid body in space from inertial sensors and relative observations of points having known locations (the setting of \cite{Vasconcelos} but with the state including the position). To our knowledge, this is the first time the invariant observer on Lie groups based approach  is applied to this full navigation problem with landmarks, except for our preliminary conference paper \cite{barrau-bonnabel-cdc14}. Indeed, this example does not fit into the usual framework leading to autonomous errors (unless we discard the position estimate as in   \cite{Vasconcelos}) but thanks to Theorem \ref{thm::CNS} we see it still leads to an autonomous error equation. This allows the IEKF observer to possess provable convergence properties. Note  the problem at hand is different from the navigation problems using magnetometers, and velocity and position measurements of the GPS  \cite{hua2010attitude,grip2015globally}. 

Of course, the EKF, to be more precise its appropriate variant the multiplicative (M)EKF \cite{lefferts1982kalman}, is the state of the (industrial) art for this navigation example, due to its good  performances, and easy tuning based on sensors' noise covariances. But to our best knowledge it is nowhere proved to possess stability properties as a non-linear observer, and  simulations below even indicate  it may diverge in some situations whereas the IEKF converges. The computations require basic formulas recalled in Appendix \ref{sect::tuto_Lie_groups}. 

\subsection{Considered  model}

%\begin{figure}[h]
%\centering
%\includegraphics[width=15cm]{Aircraft}
%\end{figure}

We consider here the  more complicated model of a vehicle evolving in the 3D space and characterized by its attitude $R_t$, velocity $v_t$ and position $x_t$. The vehicle is endowed with accelerometers and gyroscopes whose measures are denoted respectively by $u_t$ and angular velocity $\omega_t$. The   dynamics read:
\begin{equation}
\label{3D_dynamics:nonoise} 
\frac{d}{dt} R_t  = R_t (\omega_t)_\times \qquad , \qquad
\frac{d}{dt} v_t  = g + R_t u_t  \qquad , \qquad
\frac{d}{dt} x_t  = v_t ,
\end{equation}
where $(\omega)_\times$ denotes the $3\times 3$ skew-symmetric matrix associated with the cross product with $\omega$, that is, for any $b\in\RR^3$ we have $(\omega)_\times b=\omega\times b$. 
Observations of the relative position of known features (using for instance a depth camera) are considered:
\begin{equation}
(Y_n^1,...,Y_n^k) = \left( R_{t_n}^T (p_1-x_{t_n}) , ..., R_{t_n}^T (p_k-x_{t_n})   \right), \label{3D_features:nonoise}
\end{equation} where $(p_1,...,p_k)$ denote the (assumed known) position of the features in the earth-fixed frame.

\subsection{IEKF gain tuning}
To derive and tune the IEKF equations, we follow the methodology of Section \ref{ouida} which amounts to 1- associate a ``noisy "system to the original considered system, just because it allows obtaining a sensible  tuning of the design matrices from an engineering viewpoint, 2- transform it into a system defined on a matrix Lie group to make it fit into our framework, 3- linearize the ``noisy" equations, and 4- use the Kalman equations to tune the observer gain. 

\subsubsection{Associated ``noisy" system}

By merely introducing noise in the accelerometers' and gyrometers' measurements we obtain the well-known equations \cite{farrell2008aided}:
\begin{equation}
\label{3D_dynamics} 
\frac{d}{dt} R_t  = R_t (\omega_t+w_t^{\omega})_\times  , \quad
\frac{d}{dt} v_t  = g + R_t (u_t+w_t^u)   , \quad
\frac{d}{dt} x_t  = v_t .
\end{equation}
Letting additive noise pollute the observations (the sensor being in the body frame) we get:
\begin{equation}
(Y_n^1,...,Y_n^k) = \left( R_{t_n}^T (p_1-x_{t_n}) + V_n^1, ..., R_{t_n}^T (p_k-x_{t_n})  + V_n^k \right), \label{3D_features}
\end{equation}
where $V_n$, $V_n^1,...,V_n^k$ are noises in $\RR^3$.

\subsubsection{Matrix form}
As already noticed in the preliminary work \cite{barrau-bonnabel-cdc14}, the system \eqref{3D_dynamics} can be embedded in the group of double homogeneous matrices (see Appendix \ref{sect::tuto_SE23}) using the matrices $\chi_t$, $w_t$ and function $f_{\omega,u}$:
\[
\chi_t = \begin{pmatrix} R_t & v_t & x_t \\ 0_{3,1} & 1 & 0 \\ 0_{3,1} & 0 & 1 \end{pmatrix}, \qquad w_t = \begin{pmatrix} (w_t^\omega)_\times & w_t^u & 0_{3,1} \\
0_{1,3} & 0 & 0 \\
0_{1,3} & 0 & 0 \end{pmatrix},
\]
\[
f_{\omega,u}: \begin{pmatrix} R & v & x \\ 0_{3,1} & 1 & 0 \\ 0_{3,1} & 0 & 1 \end{pmatrix} \rightarrow \begin{pmatrix} R(\omega)_\times & g + Ru & v \\ 0_{3,1} & 0 & 0 \\ 0_{3,1} & 0 & 0 \end{pmatrix}.
\]
The equation of the dynamics becomes:
\begin{equation}
\frac{d}{dt} \chi_t = f_{\omega_t,u_t}(\chi_t) + \chi_t w_t, \label{eq::3D_matrix_dynamique}
\end{equation}
and the observations \eqref{3D_features} have the equivalent forms:
\begin{equation}
(Y_n^1,...,Y_n^k) = \left( \chi_{t_n}^{-1} \begin{pmatrix} p_1 \\ 0 \\ 1 \end{pmatrix} +\begin{pmatrix}  V_n^1\\ 0 \\ 0 \end{pmatrix}, ...,  \chi_{t_n}^{-1} \begin{pmatrix} p_k \\ 0 \\ 1 \end{pmatrix} +\begin{pmatrix}  V_n^k\\ 0 \\ 0 \end{pmatrix}  \right) . \label{eq::3D_matrix_features}
\end{equation}
\begin{prop}
The matricial function $f_{\omega_t, u_t}$ is neither left nor right invariant. However the reader can verify relation \eqref{eq::main_relation}  which is easy to derive.
\end{prop}

%%%%%%%%%%%%%%%%%%%%%%%%%%%%%%%%%%%%%%%%%%%%%%%%%%%
%Nav 3D + features
%%%%%%%%%%%%%%%%%%%%%%%%%%%%%%%%%%%%%%%%%%%%%%%%%%%
\subsubsection{IEKF equations}
\label{sect::IEKF_eq_3D}
The RIEKF \eqref{LIEKF:::eq} for the associated ``noisy" system
 \eqref{eq::3D_matrix_dynamique}, with right-invariant ``noisy" observations \eqref{eq::3D_matrix_features} reads:
\[
\frac{d}{dt} \hat{\chi}_t = f_{\omega_t,u_t} (\hat{\chi}_t)
\qquad , \qquad
\hat{\chi}_t^+= \exp \left( L_n \begin{pmatrix} \hat{\chi}_t Y_n^1 \\ ... \\ \hat{\chi}_t Y_n^k  \end{pmatrix} \right) \hat{\chi}_t.
\]
As the two last entries of each matrix $\hat{\chi}_t^{-1}Y_n^j$ are always zero, one we can conveniently use a reduced-dimension gain matrix $\tilde{L}_n$ defined by $L_n=\tilde{L}_n \tilde{p}$ with $\tilde{p}=\begin{pmatrix} [I_3, 0_{3,2}] & & \\ & \ddots \\ & & [I_3, 0_{3,2}] \end{pmatrix}$.  
The right-invariant error is $
\eta_t =  \hat{\chi}_t \chi_t^{-1}
$
and its evolution reads:

\begin{align}
 \frac{d}{dt} \eta_t & = f_{\omega_t,u_t} (\eta_t) - \eta_t f_{\omega_t,u_t}(I_5) - (\hat{\chi}_t  w_t \hat{\chi}^{-1}) \eta_t , \label{eq::error_3D_dynamics_features} \\
\eta_{t_n}^+ & =  \exp \left( \tilde{L}_n \tilde{p} \begin{pmatrix} \eta_{t_n} \begin{pmatrix} p_1 \\ 0 \\ 1 \end{pmatrix} +\hat{\chi}_{t_n} \begin{pmatrix} V_n^1 \\ 0_{2,1} \end{pmatrix}  \\ ... \\  \eta_{t_n} \begin{pmatrix} p_k \\ 0 \\ 1 \end{pmatrix} + \hat{\chi}_{t_n} \begin{pmatrix} V_n^k \\ 0_{2,1} \end{pmatrix}  \end{pmatrix} \right) \eta_{t_n}. \label{eq::error_3D_observation_features}
\end{align}

To linearize this equation we introduce the linearized error $\xi_t$ by replacing $\eta_t$ with $I_3 + \mathcal{L}_{\frak{se}(3),2}(\xi_t)$. Using the first order approximations $\eta_t = I_3 + \mathcal{L}_{\frak{se}(3),2}(\xi_t)$, $\eta_t^+ = I_3 + \mathcal{L}_{\frak{se}(3),2}(\xi_t^+)$, $\exp(u)=I_3+\mathcal{L}_{\frak{se}(3),2}(u)$ and $\eta_t^{-1} = I_3 - \mathcal{L}_{\frak{se}(3),2}(\xi_t) $ in \eqref{eq::error_3D_dynamics_features}, \eqref{eq::error_3D_observation_features} and removing the second-order terms in $\xi_t$, $V_n$ and $w_t$ we obtain:

\begin{align*}
\frac{d}{dt} \xi_t & = \begin{pmatrix} 0_{3,3} & 0_{3,3} & 0_{3,3} \\ (g)_\times & 0_{3,3} & 0_{3,3} \\ 0_{3,3} & I_3 & 0_{3,3} \end{pmatrix} \xi_t - \begin{pmatrix} \hat{R}_t & 0_{3,3} & 0_{3,3} \\ (\hat{v}_t)_\times \hat{R}_t & \hat{R}_t & 0_{3,3} \\ (\hat{x}_t)_\times \hat{R}_t & 0_{3,3} & \hat{R}_t \end{pmatrix} w_t, \\
\xi_{t_n}^+ & = \xi_{t_n} - \tilde{L}_n \left[ \begin{pmatrix}  (p_1)_\times & 0_{3,3} & -I_3 \\ & ... & \\ (p_k)_\times & 0_{3,3} & -I_3  \end{pmatrix} \xi_{t_n} - \begin{pmatrix} \hat{R}_{t_n} V_n^1 \\ ... \\  \hat{R}_{t_n} V_n^k \end{pmatrix} \right].
\end{align*}
The gains $\tilde{L}_n$  are computed using the Riccati equation \eqref{eq::Riccati_Left} and matrices $A_t,H, \hat{Q}$ and $\hat{N}$ defined as:
\[
A_t = \begin{pmatrix} 0_{3,3} & 0_{3,3} & 0_{3,3} \\ (g)_\times & 0_{3,3} & 0_{3,3} \\ 0_{3,3} & I_3 & 0_{3,3} \end{pmatrix} , \qquad H = \begin{pmatrix}  (p_1)_\times & 0_{3,3} & -I_3 \\ & ... & \\ (p_k)_\times & 0_{3,3} & -I_3  \end{pmatrix},
\]
\[
\hat{Q}= \begin{pmatrix} \hat{R}_t & 0_{3,3} & 0_{3,3} \\ (\hat{v}_t)_\times \hat{R}_t & \hat{R}_t & 0_{3,3} \\ (\hat{x}_t)_\times \hat{R}_t & 0_{3,3} & \hat{R}_t \end{pmatrix} Cov(w_t) \begin{pmatrix} \hat{R}_t & 0_{3,3} & 0_{3,3} \\ (\hat{v}_t)_\times \hat{R}_t & \hat{R}_t & 0_{3,3} \\ (\hat{x}_t)_\times \hat{R}_t & 0_{3,3} & \hat{R}_t \end{pmatrix}^T ,
\]
\[
\hat{N} = \begin{pmatrix} \hat{R}_{t_n} Cov(V_n^1) \hat{R}_{t_n}^T & & \\ & \ddots & \\ &  &  \hat{R}_{t_n} Cov(V_n^k) \hat{R}_{t_n}^T \\  \end{pmatrix}.
\]

\subsection{Stability properties of the IEKF viewed as a non-linear observer for the navigation example}
\begin{thm}
\label{thm::3D_features_stability}
If three non-collinear points are observed, then the  IEKF whose equations are derived in Section \ref{sect::IEKF_eq_3D} is an asymptotically stable observer for the system \eqref{3D_dynamics:nonoise}-\eqref{3D_features:nonoise}   in the sense of Definition \ref{deff} about \emph{any} bounded trajectory.
\end{thm}

\begin{proof}
According to Theorem \ref{thm::bounded} we only have to ensure the couple (A,H) is observable. Integrating the propagation on one step we obtain the discrete propagation matrix $$ \Phi=\begin{pmatrix} I_3 & 0_{3 \times 3} & 0_{3 \times 3} \\ t(g)_{\times} & I_3 & 0_{3 \times 3} \\ \frac{1}{2} t^2 (g)_\times  & t Id_{3 \times 3} & I_3 \end{pmatrix}. $$ The observation matrix is denoted $H$. We will show that $[H ; H \Phi]$ has rank $9$. We can keep only the raws corresponding to the observation of three non-collinear features $p_1,p_2,p_3$ and denote the remaining matrix by $\mathcal{H}_1$. Matrices $\mathcal{H}_2$ and $\mathcal{H}_3$, obtained using elementary operations on the columns of $\mathcal{H}_1$, have a rank inferior or equal to the rank of $\mathcal{H}_1$:
\[
\mathcal{H}_1=\begin{pmatrix}
(p_1)_\times & 0_{3 \times 3} & -I_3  \\
(p_2)_\times & 0_{3 \times 3} & -I_3  \\
(p_3)_\times & 0_{3 \times 3} & -I_3  \\
(p_1)_\times - \frac{1}{2}t^2 (g)_\times & -t I_3 & -I_3 \\
(p_2)_\times - \frac{1}{2}t^2 (g)_\times & -t I_3 & -I_3 \\
(p_3)_\times - \frac{1}{2}t^2 (g)_\times & -t I_3 & -I_3
 \end{pmatrix},
\qquad
\mathcal{H}_2=\begin{pmatrix}
(p_1)_\times & 0_{3 \times 3} & -I_3  \\
(p_2)_\times & 0_{3 \times 3} & -I_3  \\
(p_3)_\times & 0_{3 \times 3} & -I_3  \\
-\frac{1}{2}t^2 (g)_\times & -t I_3 & 0_{3 \times 3} \\
-\frac{1}{2}t^2 (g)_\times & -t I_3 & 0_{3 \times 3} \\
-\frac{1}{2}t^2 (g)_\times & -t I_3 & 0_{3 \times 3}
 \end{pmatrix},
 \]
 \[
\mathcal{H}_3=\begin{pmatrix}
(p_1-p_3)_\times & 0_{3 \times 3} & 0_{3 \times 3}  \\
(p_2-p_3)_\times & 0_{3 \times 3} & 0_{3 \times 3}  \\
\frac{1}{2}t^2 (g)_\times & t I_3 & 0_{3 \times 3} \\
(p_3)_\times & 0_{3 \times 3} & -I_3 
 \end{pmatrix}.
\]
The diagonal blocks $\begin{pmatrix}
-(p_1-p_3)_\times & 0_{3 \times 3}  \\
-(p_2-p_3)_\times & 0_{3 \times 3} 
 \end{pmatrix}$, $t I_3$ and $I_3$ have rank $3$ thus the full matrix has rank $9$.
\end{proof}

\subsection{Simulations}

The IEKF described in Section \ref{sect::IEKF_eq_3D} has been implemented and compared to a state of the art multiplicative EKF \cite{lefferts1982kalman} for the experimental setting described by Figure \ref{fig::Simu_3D} (top plots). The vehicle drives a 10-meter diameter circle (green arrows) in 30 seconds and observes three features (black circles) every second while receiving high-frequency inertial measurements (100 Hz).  The equations of the IEKF have already been detailed. The error variable to be linearized for the multiplicative (M)EKF is $e_t=(\hat R_t R_t^{-1}, \hat v_t - v_t, \hat x_t - x_t)$. As $\hat R_t R_t^{-1}$ is not a vector variable,  it is linearized using the first-order expansion  $\hat R_t R_t^{-1} = I_3 + \left( \zeta_t \right)_{\times}, \zeta_t \in \mathbb{R}^3$ \cite{lefferts1982kalman}. The linearized error variable is thus a vector $ \epsilon_t = \left( \zeta, dx, dv \right) \in \mathbb{R}^9$. Expanding the propagation and observation steps up to the first order in $\epsilon_t$ give the classical $F_t$ and $H_n$ matrices used in the Riccati Equation of the MEKF:
 $$
F_t = \begin{pmatrix}
0_3 & 0_3 & 0_3 \\
- \left( \hat R_t u_t \right)_\times & 0_3 & 0_3 \\
0_3 & I_3 & 0_3
\end{pmatrix},
$$
$$
H_n = 
\begin{pmatrix}
-\hat R_{t_n}^T \left( p^1 - \hat x_{t_n} \right)_{\times} & \quad 0_3 \quad & \quad \hat R_{t_n}^T \quad \\
 & \ldots &  \\
-\hat R_{t_n}^T \left( p^k - \hat x_{t_n} \right)_{\times} & \quad 0_3 \quad & \quad \hat R_{t_n}^T \quad
\end{pmatrix}.
$$
We use the following design parameters in two distinct simulations, with same $N$ but two different matrices $Q$.
$$
N = \begin{pmatrix}10^{-2} I_3 & 0_3 & 0_3 \\ 0_3 & 10^{-2} I_3 & 0_3 \\ 0_3 & 0_3 & 10^{-2} I_3 \end{pmatrix},
$$
$$
Q_1 = \begin{pmatrix} 10^{-8} I_3 & 0_3 & 0_3 \\ 0_3 & 10^{-8} I_3 & 0_3 \\ 0_3 & 0_3 & 0_3 \end{pmatrix},
\quad
Q_2 = \begin{pmatrix} 10^{-4} I_3 & 0_3 & 0_3 \\ 0_3 & 10^{-4} I_3 & 0_3 \\ 0_3 & 0_3 & 0_3 \end{pmatrix}.
$$
The initial errors are the same for both simulations:  15 degrees for attitude and  1 meter for position standard deviations. The small ``process noise'' matrix $Q_1$, although reasonable in the context of high-precision inertial navigation, has been deliberately chosen to challenge EKF-like methods: the corresponding gains are small so  the errors introduced during the transitory phase due to non-linearities in the initial errors can never be corrected. Note this would not be an issue if the system was linear: the estimation errors and filter gains would decrease simultaneously. The problem is that the error does not decrease as fast as predicted by the \emph{linear} Kalman theory. As shown by the plots of the left column of Figure \ref{fig::Simu_3D} (top plots) it makes the EKF even diverge! This is probably the simplest way to  make the EKF fail in a navigation problem and this is purely a problem of  non-linearity as no noise has been added  whatsoever. Still on the left column, we see the IEKF is not affected by the problem, due to its appropriate non-linear structure. In particular, the attitude and position errors go to zero in accordance with Theorem \ref{thm::3D_features_stability}.

Usually, engineers get around those convergence problems by   artificially inflating the ``process noise'' matrix $Q$ (see also \cite{reif}). This classical solution, sometimes referred to as robust tuning, is illustrated here by using   $Q_2$ instead. The results are displayed on the right column of Fig. \ref{fig::Simu_3D}. They illustrate the fact the EKF, as an observer, can be improved  through a proper tuning, although still much slower to converge that the IEKF. But this raises issues: $Q$ and $N$ have been chosen for a specific trajectory with no guarantee regarding robustness. Moreover, these matrices admit a physical interpretation (the accuracy of the sensors) and arbitrarily changing them by several  orders of magnitude is a renouncement to use this precious information when available. In turn,  this makes  the matrix $P_t$ loose its interpretability as an indication of the observer's accuracy in response to the sensors' trusted accuracy.  For this relevant problem, we thus see the IEKF turns out to be a viable alternative to the EKF thanks to its \emph{guaranteed} properties, and to its convincing experimental behavior reflecting way better performances than the EKF, even for challenging choices of $Q$ and $N$.

\begin{figure*}[!h]
  \centering
\begin{minipage}[b][]{0.44\linewidth}\noindent
\includegraphics[width=\linewidth]{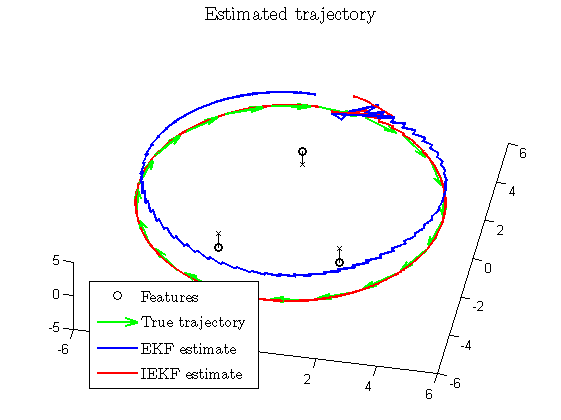}
\end{minipage}
\begin{minipage}[b][]{0.44\linewidth}\noindent
\includegraphics[width=\linewidth]{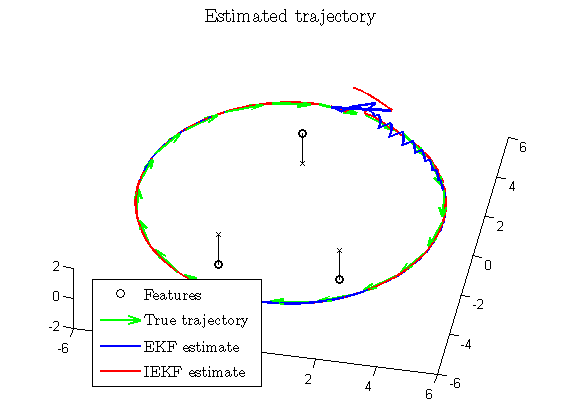}
\end{minipage}
\begin{minipage}[b][]{0.44\linewidth}
\includegraphics[width=\linewidth]{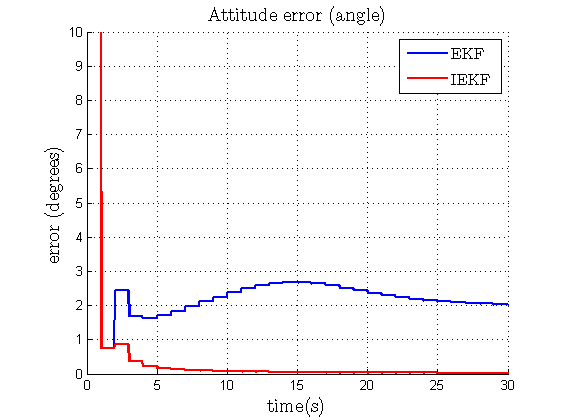}
\end{minipage}
\begin{minipage}[b][]{0.44\linewidth}
\includegraphics[width=\linewidth]{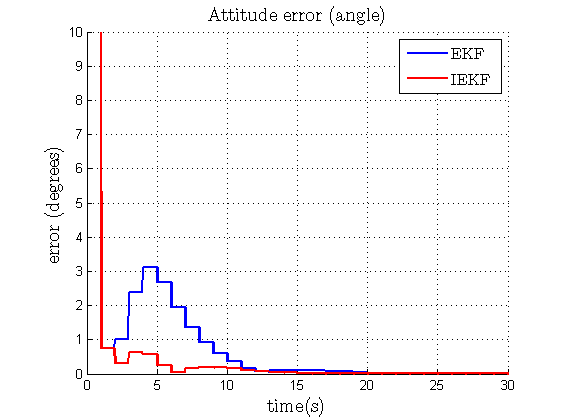}
\end{minipage}
\begin{minipage}[b][]{0.44\linewidth}\noindent
\includegraphics[width=\linewidth]{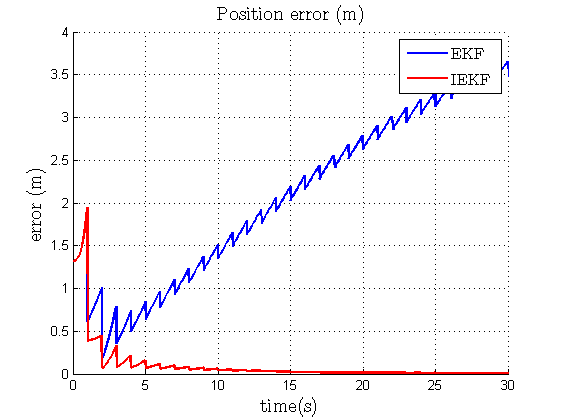}
\end{minipage}
\begin{minipage}[b][]{0.44\linewidth}\noindent
\includegraphics[width=\linewidth]{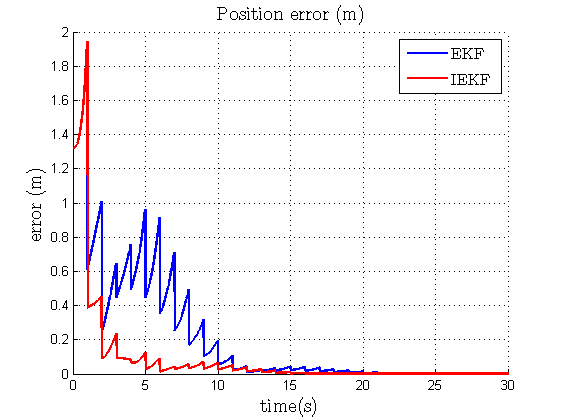}
\end{minipage}
  \caption{Aided inertial navigation based on high rate accelerometers' and gyrometers' measurements and low rate observation of known landmarks. We also displayed the orthogonal projection of the landmarks on the plane containing the trajectory (black crosses) to help imagining the 3D position of the landmarks. This shows the disposition of the landmarks is the same in both experiments. Top plots illustrate the experimental setting and display the  EKF and IEKF estimates. Middle plots display the attitude errors and bottom plots the position errors. \textbf{Left column:} the tuning of $Q$ is tight ($Q=Q_1$) due to highly precise inertial sensors. This creates robustness issues: the gains of the EKF decrease rapidly during the transitory phase while the attitude error is not reduced enough due to non-linearities. When the position estimate is impacted, the gains have become too small to correct the errors, leading to filter's divergence. IEKF ensures rapid decrease to 0 of the estimation error with identical tuning. \textbf{Right plot:}   $Q$ is inflated ($Q=Q_2$). This classical engineering trick prevents the EKF to diverge but IEKF still prevails in terms of time of convergence.}
\label{fig::Simu_3D}
\end{figure*}

\section{Conclusion}
The Invariant EKF, when used as a deterministic observer for an introduced and well characterized class of problems on Lie groups, is shown to possess theoretical stability guarantees under the simple and natural hypotheses of the linear case, a feature the EKF has never been proved to share so far. Simulations confirm the   IEKF is an appealing alternative indeed, as it is always superior to the EKF and  outperforms it in challenging situations, while remaining similar to EKF in terms of tuning, implementation, and computational load. 
%The authors had no room for simulations, which would have been anyway beyond the scope of this theoretical paper. Practical results can be found in \cite{barrau2013intrinsic}.

%%%%%%%%%%%%%%%%%%%%%%%%%%%%%%%%%%%%%%%%%%%%%%%%%
%APPENDIX
%%%%%%%%%%%%%%%%%%%%%%%%%%%%%%%%%%%%%%%%%%%%%%%%%

\appendix

 \section{Matrix Lie groups useful formulas}
\label{sect::tuto_Lie_groups}

A matrix Lie group $G$ is a set of square invertible matrices of size $N \times N$ verifying the following properties:
\[
Id \in G, \qquad \forall g \in G, g^{-1} \in G , \qquad \forall a,b \in G, ab \in G
\]

If $\gamma(t)$ is a process taking values in $G$ and verifying $\gamma(0)=Id$, then its derivative at $t=0$ cannot take any value in the set of squared matrices $\mathcal{M}_N(\RR)$. It is constrained to lie in a vector subspace $\mathfrak{g}$ of $\mathcal{M}_N( \RR )$ called the ``Lie algebra of $G$". As it proves useful to identify $\mathfrak{g}$ to $\RR^{\dim \mathfrak{g}}$, a linear mapping $\mathcal{L}_{\frak{g}}: \RR^{\dim \frak{g}} \rightarrow \frak{g}$ is used in this paper. The vector space $\mathfrak{g}$ can be mapped to the matrix Lie group $G$ through the classical matrix exponential $\exp_m$. As well, $\RR^{\dim \mathfrak{g}}$ can be mapped to $G$ through a function $\exp$ defined by $\exp(\xi)=\exp_m(\mathcal{L}_{\mathfrak{g}}(\xi))$. For any $g\in G$ the operator $Ad_g : \RR^{\dim \mathfrak{g}} \rightarrow  \RR^{\dim \mathfrak{g}}$ is defined by $g \mathcal{L}_{\mathfrak{g}}(\xi) g^{-1} = \mathcal{L}_{\mathfrak{g}}(Ad_g \xi)$. For any $x \in \RR^{\dim \mathfrak{g}}$ the operator $Ad_{x} : \RR^{\dim \mathfrak{g}} \rightarrow  \RR^{\dim \mathfrak{g}}$ is defined by $\mathcal{L}_{\mathfrak{g}}(x) \mathcal{L}_{\mathfrak{g}}(\xi) - \mathcal{L}_{\mathfrak{g}}(\xi) \mathcal{L}_{\mathfrak{g}}(x) = \mathcal{L}_{\mathfrak{g}}(Ad_x \xi)$. These operators are very handy in practical computations. For all matrix Lie groups considered in this paper, no matrix exponentiation is actually needed as there exist closed formulas, given thereafter. We give now a short description of the matrix Lie groups appearing in the present paper.

\subsection{Group of direct planar isometries $SE(2)$}
\label{sect::tuto_SE2}

We have here $G=SE(2)$ and $\mathfrak{g} = \mathfrak{se}(2)$, where $$
SE(2)= \left\lbrace \begin{pmatrix} R(\theta) & x \\ 0_{1,2} & 1 \end{pmatrix}, \begin{pmatrix} \theta \\ x \end{pmatrix} \in \RR^3 \right\rbrace , $$ $$ { \frak{se}(2)= \left\lbrace \begin{pmatrix} 0 & -\theta & u_1 \\ \theta & 0 & u_2 \\ 0 & 0 & 0 \end{pmatrix}, \begin{pmatrix} \theta \\ u_1 \\ u_2 \end{pmatrix} \in \RR^3 \right\rbrace, }$$ $$ { \mathcal{L}_{\mathfrak{se}(2)} \begin{pmatrix} \theta \\ u_1 \\ u_2 \end{pmatrix} = \begin{pmatrix} 0 & -\theta & u_1 \\ \theta & 0 & u_2 \\ 0 & 0 & 0 \end{pmatrix},}
$$ where
$R(\theta)=\begin{pmatrix} \cos(\theta) &  -\sin(\theta) \\ \sin(\theta) & \cos(\theta) \end{pmatrix}$ denotes the rotation of angle $\theta$, and the exponential mapping is:
\[
\exp \left( (\theta , u_1 , u_2)^T \right) = \begin{pmatrix}R(\theta) & x \\ 0_{1,2} & 1 \end{pmatrix},
\]
where
\[
x=
\begin{pmatrix}
\frac{\sin(\theta)}{\theta}          &     -\frac{1-\cos(\theta)}{\theta} \\
\frac{1-\cos(\theta)}{\theta}      &   \frac{\sin(\theta)}{\theta}   
\end{pmatrix}.
\]

\subsection{Group of double direct spatial isometries $SE_2(3)$}
\label{sect::tuto_SE23}

We have here: $$G=SE_2(3)=\left\lbrace \begin{pmatrix} R & v & x \\ 0_{1,3} & 1 & 0 \\ 0_{1,3} & 0 & 1 \end{pmatrix}, R \in SO(3), v, x \in \RR^3  \right\rbrace , $$ $$\frak g=\frak{se}_2(3)= \left\lbrace \begin{pmatrix} (\xi)_{\times} & u & y \\ 0_{1,3} & 0 & 0 \\ 0_{1,3} & 0 & 0 \end{pmatrix}, \xi, u, y \in \RR^3  \right\rbrace .$$ An isomorphism between $\RR^9$ and $\frak{se}_2(3)$ is given by $ { \mathcal{L}_{\mathfrak{se}_2(3)} \begin{pmatrix} \xi \\ u \\ y \end{pmatrix} = \begin{pmatrix} (\xi)_{\times} & u & y \\ 0_{1,3} & 0 & 0 \\ 0_{1,3} & 0 & 0 \end{pmatrix}}$. The exponential mapping is given by the formula:
\begin{align*}
\exp & \left( \left( \xi^T , u^T , y^T \right)^T \right) \\
 & = I_5 + S + \frac{1 - \cos(||\xi||)}{||\xi||^2}) S^2 + \frac{ ||\xi|| -\sin(||\xi||)}{||\xi||^3} S^3, 
\end{align*}
where  $S = \mathcal{L}_{\mathfrak{se}_2(3)}(\xi ,u , y )^T$.

%%%%%%%%%%%%%%%%%%%%%%%%%%%%%%%%%%%%%%%%%%%%%%%%
% Proof CV locale 03/10/2014
%%%%%%%%%%%%%%%%%%%%%%%%%%%%%%%%%%%%%%%%%%%%%%%%

\section{Further explanation and proof of the log-linear property}

The definition of $A_t$ through  $g_{u_t}(\exp(\xi)) = \mathcal{L}_{\mathfrak{g}}(A_{u_t} \xi) + O(\norm{\xi_t}^2)$ is very convenient in practice as shown in Sections \ref{sect::examples:A} and \ref{sect::examples:B}, but it requires a quick theoretical explanation  as follows. 
On an abstract Lie group, vector $g_{u_t} \left( \exp\left( \xi \right) \right)$ belongs to $T_{\exp \left( \xi \right)} G$ (tangent space at $\exp \left( \xi \right)$) whereas the image of $\mathcal{L}_{\mathfrak{g}}$ is $\mathfrak{g}=T_{Id}G$. But as $g_{u_t} \left( Id \right)=0$ it is true that the first order terms of $g_{u_t} \left( \exp\left( \xi \right) \right)$  belong to $\mathfrak g\subset\RR^{N\times N}$. Indeed as $\exp\left( \xi \right)^{-1}g  \left( \exp\left( \xi \right) \right)\in\mathfrak g$ we have $g  \left( \exp\left( \xi \right) \right)=\exp_m(\mathcal{L}_{\mathfrak{g}}(\xi))\mathcal{L}_{\mathfrak{g}}(\omega(\xi))$ for some function $\omega$ with $\omega(0) =0$. Thus $g  \left( \exp\left( \xi \right) \right)=(Id+\mathcal{L}_{\mathfrak{g}}(\xi))\mathcal{L}_{\mathfrak{g}}(\omega(\xi))+O(\norm{\xi}^2)=\mathcal{L}_{\mathfrak{g}}(\omega(\xi))+O(\norm{\xi}^2)$.

%First, note that as we have $g_{u_t}(Id)\equiv 0$, the Lie derivative of the vector field $g_{u_t}$ along any other  vector field only depends on the value of the vector field at $Id$. It can thus be identified to a linear operator (thus a matrix) in the tangent space at $Id$, justifying the employed notation $Dg_{u_t}(Id)$.
The proof of Theorem \ref{dechire:thm} is based upon the following lemmas.
\begin{lem}Consider the system \eqref{eq::mult_linear} and let $\bar{\chi}_t$ denote a particular solution. Consider the condition
\begin{equation}\begin{aligned}
\label{eq::mult_homogeneous}
\forall (u,a,b)\in G\quad
g_u(ab)  =ag_u(b)+g_u(a)b
\end{aligned}
\end{equation}
We have the following properties:
\begin{itemize}
\item The function $g^R_u(\eta)=f_{u_t}(\eta)-\eta f_{u_t}(Id)$  verifies \eqref{eq::mult_homogeneous} and all the solutions of \eqref{eq::mult_linear} have the form $\chi_t=\eta_t^R \bar{\chi}_t$, where $\eta_t^R$ verifies $\frac{d}{dt} \eta_t^R = g_{u_t}^R(\eta_t^R)$.
\item The function $g^L_{u_t}(\eta)=f_{u_t}(\eta)- f_{u_t}(Id)\eta$  verifies  \eqref{eq::mult_homogeneous} and all the solutions of \eqref{eq::mult_linear} have the form $\chi_t= \bar{\chi}_t\eta_t^L$, where $\eta_t^L$ verifies $\frac{d}{dt} \eta_t^L = g_{u_t}^L(\eta_t^L)$.
\end{itemize}
\end{lem}
The verification of these two properties is trivial. The functions $g_{u_t}$ governing the errors propagation turn out to possess an intriguing property. 
\begin{lem}
\label{prop::product}
Let $\Phi_t$ be the flow (that is the solution at time $t$ associated to a given initial condition) associated to the system $\frac{d}{dt}\eta_t=g_{u_t}(\eta_t)$, where $g_{u_t}$ verifies \eqref{eq::mult_homogeneous}. Then:
\[
\forall \eta_0, \eta_0' \in G, \Phi_t(\eta_0 \eta_0')=\Phi_t(\eta_0) \Phi_t(\eta_0')
\]
\end{lem}
\begin{proof}
We simply have to see that $\Phi_t(\eta_0) \Phi_t(\eta_0')$ is solution of the system $\frac{d}{dt}\eta_t=g_{u_t}(\eta_t)$:
\begin{align*}
\frac{d}{dt} [\Phi_t(\eta_0) \Phi_t(\eta_0')]  & = g_{u_t}(\Phi_t(\eta_0)) \Phi_t(\eta_0') + \Phi_t(\eta_0) g_{u_t}(\Phi_t(\eta_0')) \\&= g_{u_t} \left( \Phi_t(\eta_0) \Phi_t(\eta_0') \right)
\end{align*}
% The function $g_{u_t}$ being Lipchitz, the solution of the differential equation is unique.
\end{proof}
An immediate recursion gives then:
%%%%%%%%%%%%%%%%%%%%%%%%%%%%%%%%%%%%%%%%%%%%%%%%%%%%%%
% THEOREM
%%%%%%%%%%%%%%%%%%%%%%%%%%%%%%%%%%%%%%%%%%%%%%%%%%%%%%

\begin{lem}
\label{prop::power}
We have furthermore
\[
\forall \eta_0 \in G, \forall p \in \mathbb{Z}, \Phi_t(\eta_0^p)=\Phi_t(\eta_0)^p
\]
\end{lem}
Lemmas \ref{prop::product} and \ref{prop::power} indicate the behavior of the flow infinitely close to $Id$ dictates its behavior arbitrarily far from it, as the flow commutes with exponentiation. The use of the exponential thus allows deriving an infinitesimal version of the Lemma \ref{prop::power}, which is an equivalent formulation of Theorem \ref{dechire:thm}.

%%%%%%%%%%%%%%%%%%%%%%%%%%%%%%%%%%%%%%%%%%%%%%%%%%%%%%
% THEOREM
%%%%%%%%%%%%%%%%%%%%%%%%%%%%%%%%%%%%%%%%%%%%%%%%%%%%%%

\begin{thm}
\label{thm::exp}
Let $\Phi_t$ be the flow associated to the system $\frac{d}{dt}\eta_t=g_{u_t}(\eta_t)$ satisfying \eqref{eq::mult_homogeneous}.
%As we have $g_{u_t}(Id)=0$, the Lie derivative of $g_{u_t}$ along a vector field $\nu$ verifying $\nu(Id)=\xi_0$ only depends on $\xi_0$ and will be denoted $Dg_{u_t} \xi_0$.
We have:
\color{black}
\[
\Phi_t(\exp(\xi_0))=\exp(F_t \xi_0)
\]
where $F_t$ is the solution of the matrix equation $F_0 = Id, \frac{d}{dt} F_t = A_t F_t$ \color{black}.
\end{thm}
 \begin{proof}
Thanks to Lemma \ref{prop::power} we have, for any $n \in \mathbb{N}$, $
\Phi_t \left( e^{\xi_0} \right)  = \Phi_t \left( \left[ e^{\frac{1}{n} \xi_0} \right]^n \right)  = \Phi_t \left( e^{\frac{1}{n} \xi_0} \right)^n = \left[ e^{\frac{1}{n} F_t \xi_0 + r_t \left( \frac{1}{n} \xi_0 \right) } \right]^n$, where $F_t = \left.D \left( \exp^{-1} \circ \Phi_t \circ \exp \right)\right|_0$ and $r_t \left( \frac{1}{n} \xi_0 \right) $ is a quadratic term, which ensures in turn $\Phi_t \left( e^{\xi_0} \right)= e^{ F_t \xi_0 + n r_t \left( \frac{1}{n}\xi_0 \right)}$.  
Letting $n \rightarrow \infty$ we get $\Phi_t \left( e^{\xi_0} \right) =  e^{ F_t \xi_0}$. Differentiating both sides of the latter equality we obtain $g_{u_t} \left( \Phi_t \left( e^\xi_0 \right) \right)=D \exp_{F_t \xi_0} \left[ \frac{d}{dt} F_t \xi_0 \right]$. A first-order expansion in $\xi_0$ using matrix $A_t$ gives: $\mathcal{L}_{\mathfrak{g}} \left( A_t F_t \xi_0 \right) + \circ \left( || \xi_0 || \right)=\mathcal{L}_{\mathfrak{g}} \left( \frac{d}{dt}F_t \xi_0 \right) + \circ \left( || \xi_0 || \right)$, then $A_t F_t \xi_0 = \frac{d}{dt} F_t \xi_0$ for any $\xi_0 \in \mathbb{R}^{\dim \mathfrak{g}}$ and finally $A_t =  \frac{d}{dt} F_t$.
\color{black}
\end{proof}
%
%
%
%Gathering all the results of this section we obtain:
%
%The introductory example can now be interpreted as follows. The considered system can be embedded into the matrix Lie group $SE(2)$ as shown in \ref{sect::simplified_car} and the error variable $\psi$ is simply the inverse of the exponential mapping of $SE(2)$ (defined in\ref{}), applied to the left-invariant error variable.

\section{Proof of theorem \ref{thm::localCV}}
\label{proof::localCV}

%%%%%HEY HEY

\subsection{Proof rationale}
We define the rest $r_n$, here for the left-invariant filter only, as follows: $\exp \left[ (I-L_n H) \xi_n - r_n(\xi) \right]$    ~~= $\exp \left[ \xi \right] \exp \left[ -L_n H (\exp [ \xi] b - b) \right]$. We then introduce the flow $\Psi_{t_0}^t$ of the linear part  of the equations governing $\xi_t$ (that is, $
\Psi_{t_0}^{t_0} = Id , \quad 
\frac{d}{dt} \Psi_{t_0}^t = A_{u_t} \Psi_{t_0}^t ,\quad
\Psi_{t_0}^{t+} = (I-L_nH)\Psi_{t_0}^t$) and decompose the solution $\xi_t$ as:
\begin{equation}
\forall t\geq 0,\qquad\xi_t = \Psi_{t_0}^t \xi_0 + \sum_{t_n<t}\Psi_{t_n}^t r_n(\xi_{t_n}) \label{eq::decomposition}
\end{equation}
All we have to verify is that the appearance of the second-order terms $r_n(\xi_{t_n}) $ at each update is compensated by the exponential decay of $\Psi_{t_0}^t$ (Theorem \ref{thm::Deyst_Price}).
%%%%%%
\subsection{Review of existing linear results}

Consider a linear time-varying Kalman filter and let $\Psi_{t_0}^t$ denote the flow of the error variable $\xi_s$. It is proved in \cite{deyst} that if the parameters of the Riccati equation verify conditions \eqref{cond::Psi} - \eqref{cond::observable} then there exist $\gamma_{\max}>0$ and $\gamma_{\min} >0$ such that $\gamma_{\max} I \succeq P_t \succeq \gamma_{\min} I$. This pivotal property allows proving the solution of the \emph{linear} error equation $\Psi_s^t \xi_s$ verifies for $V(P,\xi)=\xi^T P^{-1} \xi$:
\begin{equation}
V \left( P_{t_{n+N}}^+,\Psi_{t_n}^{t_{n+N}}(\xi_{t_n}^+) \right) \leq V(P_{t_{n}}^+,\xi_{t_n}^+) - \beta_3 ||\Psi_{t_n}^{t_{n+N}}(\xi_{t_n}^+)||^2 \label{eq::V_decreases_exponentially}
\end{equation}
where $\beta_3$ only depends on $\alpha_1,\alpha_2,\beta_1,\beta_2,\delta_1,\delta_2,\delta_3,N$. Of course, the proof given in \cite{deyst} holds if the inequalities are only verified on an interval $[0,T]$. We will also use the  direct consequence:
\begin{equation}
V \left( P_{t_{n+1}}^+,\Psi_{t_n}^{t_{n+1}}(\xi_{t_n}^+) \right) < V \left( P_{t_n}^+,\xi_{t_n}^+ \right) \label{eq::V_decreases}
\end{equation}

\subsection{Preliminary lemmas}
The proof of Theorem \ref{thm::localCV} is displayed in the next subsection. It relies on the final Lemma \ref{lem::Lyapounov}, which is proved step by step in this section through lemmas \ref{lem::Riccati}, \ref{lem::first_order} and \ref{lem::second_order}. The time  interval between two successive observations will be denoted $\Delta t$. $\tilde P_t$ will denote the Kalman covariance about the true state trajectory.

\begin{lem}
\label{lem::Riccati} \emph{[modified constants for closeby trajectories]}
If the conditions \eqref{cond::Psi} to \eqref{cond::observable} are satisfied about the \emph{true} trajectory, then for any $k>1$ there exists a radius $\epsilon$ such that the bound $\forall s \in [0,t], || \xi_{t_0+s} || < \epsilon$ ensures the conditions \eqref{cond::Psi} to \eqref{cond::observable} are also verified on $[t_0,t_0+t]$ about the \emph{estimated} trajectory, with the modified constants $\hat{\delta}_1=\delta_1,\hat{\delta}_2=\frac{1}{k^2}\delta_2,\hat{\delta}_3=\frac{1}{k^2}\delta_3,\hat{\alpha}_1=\frac{1}{k^2}\alpha_1,\hat{\alpha}_2=k^2\alpha_2,\hat{\beta}_1= \frac{1}{k^2} \beta_1,\hat{\beta}_2=k^2 \beta_2$.  Moreover, if $\frac{1}{k} \tilde{P}_{t_0} \leq P_{t_0} \leq k \tilde{P}_{t_0}$ then $\frac{1}{k} \tilde{P}_{t_0+s} \leq P_{t_0+s} \leq k \tilde{P}_{t_0+s}$ holds on $[0,t]$.
\end{lem}

\begin{proof}
We consider the LIEKF, the proof works the same way for the RIEKF. Matrices $\hat Q_t$ and $\hat N_n$ depend on the estimate $\hat \chi_t$, this is why this lemma is needed. So we replace them by their values: $\hat Q_t = Cov \left( w_t \right)$ if the noise term has the form $\chi_t w_t$, $\hat Q_t = Ad_{\hat \chi_t^{-1}} Cov \left( w_t \right) Ad_{\hat \chi_t^{-1}}^T$ if the noise term has the form $w_t \chi_t$, and $\hat N_n = \hat \chi_{t_n}^{-1} Cov \left( V_n \right)  \hat \chi_{t_n}^{-T} + Cov \left( B_n \right)$. All these situations are covered if we assume there exist four (possibly time-dependent) matrices $Q_1$, $Q_2$, $N_1$ and $N_2$ such that $\hat Q_t =  Q_1 + Ad_{\hat \chi_t^{-1}} Q_2 Ad_{\hat \chi_t^{-1}}^T$ and $\hat N_n =  N_1+ \hat \chi_{t_n}^{-1} N_2 \hat \chi_{t_n}^{-T} $. These notations will be used in the sequel but they hold only for this proof: they are not related to matrices $Q_1$ and $Q_2$ defined in the simulations sections. The Riccati equation computed about the true trajectory reads:
\begin{align*}
&\frac{d}{dt} \tilde{P}_t  = A_t \tilde{P}_t + \tilde{P}_t A_t^T + Q_1 + Ad_{\chi_t} Q_2 Ad_{\chi_t}^T, \\
&\tilde{P}_{t_n}^+  =\tilde{P}_{t_n}-\tilde{P}_{t_n} H^T \left( H \tilde{P}_{t_n} H^T + N_1+\chi_{t_n}^{-1} N_2 \chi_{t_n}^{-T} \right)^{-1} H \tilde{P}_{t_n}.
\end{align*}
The Riccati equation computed on the estimated trajectory is obtained replacing $\chi_t$ with $\hat{\chi}_t$.
Recalling the error $\eta_t$ and the properties of the $Ad$, the idea of the proof is simply to rewrite the Riccati equation computed about $\hat{\chi}_t$ as a perturbation of the Riccati equation computed about $\chi_t$:
\[
\frac{d}{dt} P_t  = A_t P_t + P_t A_t^T + Q_1 + Ad_{\eta_t^{-1}} \left[ Ad_{\chi_t^{-1}} Q_2 Ad_{\chi_t^{-1}}^T \right] Ad_{\eta_t^{-1}}^T,
\]
\begin{align*}
  P^+_{t_n} =
   P_{t_n}-P_{t_n} H^T \left( H P_{t_n} H^T + N_1+\eta_{t_n}^{-1} \left[ \chi_{t_n}^{-1} N_2 \chi_{t_n}^{-T} \right] \eta_{t_n}^{-T} \right)^{-1} H P_{t_n}.
\end{align*}
Controlling the perturbation is easy: matrix-valued functions $\xi \rightarrow e^{-\xi}$ and $\xi \rightarrow Ad_{e^-\xi}$ are continuous and equal to $I_d$ for $\xi=0$, thus there exists a real
% let $L_\xi : x \rightarrow e^{-\xi} x$ and $A_\xi : \RR^{\dim \mathfrak{g}} \rightarrow \RR^{\dim \mathfrak{g}}, x %\rightarrow Ad_{e^{-\xi}} x$. As these functions are continuous and equal to $I_d$ for $\xi=0$ there exists a real 
$\epsilon>0$ depending only on $k$  such that $|| \xi_{t_0+s} || \leq \epsilon$ ensures $ \frac{1}{{k}} N_2 \preceq  \left( e^{-\xi_{t_0+s}} \right) N_2 \left( e^{-\xi_{t_0+s}} \right)^T \preceq {k} N_2$ and $\frac{1}{{k}} Q_2 \preceq Ad_{e^{-\xi_{t_0+s}}} Q_2 Ad_{e^{-\xi_{t_0+s}}}^T \preceq {k} Q_2$. It ensures consequently $$ \frac{1}{k} \left( Q_1 + Ad_{\chi_{t_0+s}^{-1}} Q_2 Ad_{\chi_{t_0+s}^{-1}}^T \right) \preceq Q_1 + Ad_{\hat{\chi}_{t_0+s}^{-1}} Q_2 Ad_{\hat{\chi}_{t_0+s}^{-1}}^T \preceq  k \left( Q_1 + Ad_{\chi_{t_0+s}^{-1}} Q_2 Ad_{\chi_{t_0+s}^{-1}}^T \right) $$ and $$ \frac{1}{k} \left( N_1 + \chi_{t_0+s}^{-1} N_2 \chi_{t_0+s}^{-T} \right) \preceq N_1 + \hat \chi_{t_0+s}^{-1} N_2\hat \chi_{t_0+s}^{-T} \preceq  k \left( N_1 + \chi_{t_0+s}^{-1} N_2 \chi_{t_0+s}^{-T} \right), $$ and a mere look at the definitions of the constants of Theorem \ref{thm::Deyst_Price} yields the modified constants.

The inequality $\frac{1}{k} \tilde{P}_{t_0+s} \leq P_{t_0+s} \leq k \tilde{P}_{t_0+s}$ follows from the matrix inequalities above on the covariance matrices, by writing the Riccati equation verified by $k P_t$ and $\frac{1}{k} P_t$ and using simple matrix inequalities. 

\end{proof}

\begin{lem}\emph{[first-order control of growth]}
Under the same conditions as in Lemma \ref{lem::Riccati} (including $\frac{1}{k} \tilde{P}_{t_0} \leq P_{t_0} \leq k \tilde{P}_{t_0}$) and $||\xi_{t_0+s}||$ bounded by the same $\epsilon$ for $s \in [0,2M \Delta T]$ (i.e. over $2M$ time steps, where $M$ is defined as in Theorem \ref{thm::Deyst_Price}), there exists a continuous function $l_1$ depending \emph{only} on $k$ ensuring $ || \xi_{t_0+s} || \leq l_1( || \xi_{t_0} ||)$ for any $s \in [0,2M \Delta T] $ and $\l_1 (x) = O ( x)$.
\label{lem::first_order}
\end{lem}
\begin{proof}
Using Lemma \ref{lem::Riccati} and then Theorem \ref{thm::Deyst_Price} we know there exist two constants $\gamma_{\min}>0$ and $\gamma_{\max}>0$ such that $\gamma_{\max} I \succeq P_t \succeq \gamma_{\min} I$. The non-linear rest $r_{t_n}(\xi)$ introduced in \eqref{eq::decomposition} is defined by $\exp(\xi)\exp(L_n(e^\xi b - b - H \xi ))=\exp((I-L_n H)\xi+r_{t_n}(\xi))$. The Baker-Campbell-Hausdorff (BCH) formula gives $r_{t_n}(\xi)=O(||\xi||.||L_n \xi||)$ but $L_n$ is uniformly bounded over time by $\frac{\gamma_{\max} ||H||}{\delta_3}$ as an operator. Thus $||r_n||$ is \emph{uniformly} dominated over time by a second order: there exists a continuous function $\tilde{l}^k$ (depending only on $k$ and on the \emph{true} trajectory) such that $\tilde{l}^k(x) = O(x^2)$ and $||r_{t_n} (\xi)|| \leq \tilde{l}^k (|| \xi ||)$ for any $n$ such that $t_n\leq 2M\Delta T$.

Now we can control the evolution of the error using $\tilde{l}^k$. The propagation step is linear, thus we have the classical result $\frac{d}{dt} V_t(\xi_t) <0$. It ensures $||\xi_{t_0+s}|| < \sqrt{ \frac{\gamma_{\max}}{\gamma_{\min}} } ||\xi_{t_0}||$ as long as there is no update on $[t_0,t_0+s]$. At each update step we have $V_{t_n}^+(\xi_{t_n}^+)^{1/2} = V_{t_n}^+ \left( [I-L_n H] \xi_{t_n} + r_n \left( \xi_{t_n} \right) \right)^{1/2} $ $ \leq V_{t_n}^+ \left( [I-L_n H] \xi_{t_n} \right)^{1/2} + V_{t_n}^+ \left( r_n(\xi_{t_n}) \right)^{1/2} \leq V_{t_n}\left(\xi_{t_n} \right)^{1/2} + V_{t_n}^+ \left( r_n(\xi_{t_n} ) \right)^{1/2}$ using the triangular inequality. Thus: $||\xi_{t_n}^+|| \leq \sqrt{ \frac{\gamma_{\max}}{\gamma_{\min}} } \left( || \xi_{t_n} || + ||r_n (\xi_{t_n}) || \right) $~~$\leq \sqrt{ \frac{\gamma_{\max}}{\gamma_{\min}} } \left( || \xi_{t_n} || +\tilde{l}^k( || \xi_{t_n} ||) \right)$. Reiterating over successive propagations and updates  over $[0, 2 M \Delta T]$,  we see  $||\xi_{t_0+s}||$ is uniformly bounded by a  function $l_1(||\xi_{t_0}||)$ that is first order in $||\xi_{t_0}||$.

\end{proof}

\begin{lem}\emph{[second-order control of the Lyapunov function]}
Under the same conditions as in Lemma \ref{lem::Riccati} (including $\frac{1}{k} \tilde{P}_{t_0} \leq P_{t_0} \leq k \tilde{P}_{t_0}$)  and for $||\xi_{t_0+s}||$ bounded by the same $\epsilon$ for $s \in [0, 2M \Delta T]$ ($2M$ time steps, see Theorem \ref{thm::Deyst_Price} for the definition of $M$), there exists a continuous function $l_2$ depending \emph{only} on $k$ ensuring $V_{t_0+s} \left( \xi_{t_0+s} \right) \leq V_{t_0+s} (\Psi_{t_0}^{t_0+s}\xi_{t_0}) + l_2( || \xi_{t_0} ||) \leq V_{t_0} (\xi_{t_0}) + l_2( || \xi_{t_0} ||) $ for any $s \in [0,2 M \Delta t] $ with $\l_2 (x) = O ( x^2)$. We also have $V_{t_n}\left( \xi_{t_n}^+ \right) \leq V_{t_n}^+ ( (\Psi_{t_0}^{t_n})^+ \xi_{t_0}) + l_2( || \xi_{t_0} ||) \leq V_{t_0} (\xi_{t_0}) + l_2( || \xi_{t_0} ||)$ for $t_0 \leq t_n \leq t_0 + 2M \Delta T$.
\label{lem::second_order}
\end{lem}
\begin{proof}
The result stems from the decomposition \eqref{eq::decomposition} as:
\begin{align*}
 & V_{t_0+s}  (\xi_{t_0+s})^{1/2} = V_{t_0+s} \left( \Psi_{t_0}^{t_0+s} \xi_{t_0} + \sum_{t_0<t_n<t_0+s} \Psi_{t_n}^{t_0+s} r_n(\xi_{t_n}) \right)^{1/2} \\
 & \leq V_{t_0+s} \left( \Psi_{t_0}^{t_0+s} \xi_{t_0+s} \right)^{1/2} + \sum_{t_0<t_n<t_0+s} V_{t_0+s} \left( \Psi_{t_n}^{t_0+s} r_n(\xi_{t_n}) \right)^{1/2} \\
& \tag*{ (triangular inequality)}\\
 & \leq V_{t_0+s}  \left( \Psi_{t_0}^{t_0+s} \xi_{t_0+s} \right)^{1/2} + \sum_{t_0<t_n<t_0+s} V_{t_n} \left( r_n(\xi_{t_n}) \right)^{1/2} \tag*{(using \eqref{eq::V_decreases})}\\
 & \leq V_{t_0+s}  \left( \Psi_{t_0}^{t_0+s} \xi_{t_0+s} \right)^{1/2} + \sum_{t_0<t_n<t_0+s} \sqrt{ \frac{\gamma_{\max}}{\gamma_{\min}} } || r_n(\xi_{t_n}) ||^{1/2} \\
 & \tag*{(from the def. of $V$)}\\
 & \leq V_{t_0+s}  \left( \Psi_{t_0}^{t_0+s} \xi_{t_0+s} \right)^{1/2} + \sum_{t_0<t_n<t_0+s} \sqrt{ \frac{\gamma_{\max}}{\gamma_{\min}} } ( \tilde{l}^k \circ l_1^k ) (||\xi_{t_0}||)^{1/2} \\
 & \tag*{(from Lemma \ref{lem::first_order}).}
\end{align*}
As we have $( \tilde{l}^k \circ l_1^k )(x) = O(x^2)$, we obtain the result squaring the inequality and using $V_{t_0+s}  \left( \Psi_{t_0}^{t_0+s} \xi_{t_0+s} \right)$ $ \leq V_{t_0}  \left( \xi_{t_0} \right) \leq \frac{\gamma_{\max}}{\gamma_{\min}} ||\xi_{t_0} ||$ to control the crossed terms.
\end{proof}

\begin{lem}\emph{[final second order growth control]}
Under the same conditions as in Lemma \ref{lem::Riccati} (including $\frac{1}{k} \tilde{P}_{t_0} \leq P_{t_0} \leq k \tilde{P}_{t_0}$)  and for $||\xi_{t_0+s}||$ bounded by the same $\epsilon$ for $s \in [0,t]$, there exist two functions $l_1^{k}(\xi) = O(||\xi||^2)$ and $l_2^{k}=o(||\xi||^2)$ and a constant $\beta^{k}$ ensuring the relation:
\begin{align}
V_{t_0+s}(\xi_{t_0+s}) & \leq V_{t_0}(\xi_{t_0}) + l_1^{k} (\xi_{t_0}) \\
 & -  \sum_{i=0}^{J-1} \left[ \beta^{k} ||\xi_{t_{ n_0}+iM \Delta t}||^2 - l_2^k (\xi_{t_{ n_0}+iM \Delta t}) \right]
+ l_1^k(||\xi_{t_{n_{\max}}}^+||),
\label{eq::Lyapounov}
\end{align}
where ${n_{\max}}$ is the last update before $t_0+s$ ( i.e. ${n_{\max}}  = \max \left\lbrace n, t_n \leq t_0+s \right\rbrace$), $J$ is the number of successive sequences of $M$ updates in $[t_0+M \Delta t,t_{n_{\max}}]$ (i.e. $J = \max \left\lbrace j, t_{n_{\max}-jM} \geq t_0 \right\rbrace -1 $) and $n_0 = n_{\max} - JM$. If $t_0+s=t_{n_{\max}}$ the last term can be removed.
\label{lem::Lyapounov}
\end{lem}

\begin{proof}
For $l_1^k$ we choose the same function as in Lemma \ref{lem::second_order}. There is nothing more to prove for $s < 2 M \Delta t$. Let $s \geq 2 M \Delta  t$. We have $V_{t_0+s}(\xi_{t_0 + s})-V_{t_0}^+(\xi_{t_0}^+) = \left( V_{t_0+s}(\xi_{t_0 + s})-V_{t_{n_{\max}}}^+ (\xi_{t_{n_{\max}}}^+) \right)  + \left( V_{t_{n_{\max}}}^+ (\xi_{t_{n_{\max}}}^+) - V_{t_{n_0}}^+ (\xi_{t_{n_0}}^+) \right) + \left( V_{t_{n_0}}^+ (\xi_{t_{n_0}}^+) - V_{t_0} (\xi_{t_0}) \right) $. The first and third terms are upper bounded using Lemma \ref{lem::second_order}. The second term is controlled as follows:
\begin{align*}
 & V_{t_{n_{\max}}}^+  (\xi_{t_{n_{\max}}}^+) - V_{t_{n_0}}^+ (\xi_{t_{n_0}}^+) \\
 & =  \sum_{i=0}^{J-1} \left[ V_{t_{ n_{0} +  (i+1)M }}^+ (\xi_{t_{ n_{0} +  (i+1)M}}^+) - V_{t_{ n_0+iM}}^+ (\xi_{t_{ n_0 + iM}}^+) \right] \\
& \leq \sum_{i=0}^{J-1} \left[ V_{t_{ n_0+(i+1)M} }^+ (\Psi_{t_{ n_0+iM}}^{t_{ n_0+(i+1)M}+} \xi_{t_{n_0 +  iM}}^+) - V_{t_{ n_0+iM} }^+ (\xi_{t_{ n_0 +  iM }}^+) + l_2 (\xi_{t_{ n_0+iM}}^+)
\right] .
\end{align*}
And we conclude using (see \cite{deyst}): 
\begin{align*}
 V_{t_{ n_0+(i+1)M} }^+ & (\Psi_{t_{ n_0+iM}}^{t_{ n_0+(i+1)M}+} \xi_{t_{n_0 +  iM}}^+) - V_{t_{ n_0+iM} }^+ (\xi_{t_{ n_0 +  iM }}^+) \leq \\ & - \tilde{\beta}^{k} || \Psi_{t_{ n_0+iM}}^{t_{ n_0+(i+1)M}+} \xi_{t_{n_0 +  iM}}^+ ||^2 \leq - \tilde{\beta}^{k} (\frac{\gamma_{\min}}{\gamma_{\max}} \delta_1)^M || \xi_{t_{n_0 +  iM}}^+ ||^2 ,
\end{align*}
for a $\tilde{\beta}^k$ depending only on the modified constants of Lemma \ref{lem::Riccati}. The last inequality is obtained using $\Psi_{t_0}^{t_n +} = (P_n^+ P_n^{-1}) \Psi_{t_0}^{t_n}$ and an obvious recursion over $M$ time steps. We finally set $\beta = \tilde{\beta}^{k} (\frac{\gamma_{\min}}{\gamma_{\max}} \delta_1)^M$.
\end{proof}

\begin{rem}
The control we have obtained on $\xi_{t_0+s}$ is verified if $||\xi_{t_0+s}||$ is \emph{already} in a ball of radius $\epsilon$ over the whole interval $[t_0, t_0 + t]$. We now prove the result holds assuming \emph{only} that $\xi_{t_0}$ is sufficiently small.
\end{rem}

\subsection{Proof of theorem \ref{thm::localCV}}
Applying Lemma \ref{lem::Lyapounov} with $t_0+s=t_{n_{\max}}$ gives for $\frac{1}{k} \tilde{P}_{t_0} \leq P_{t_0} \leq k \tilde{P}_{t_0}$ and $||\xi_{t_0+s}|| < \epsilon$ on $[0,t]$ :
\begin{align*}
|| \xi_{t_{n_{\max}}}^+ ||^2 \leq & \frac{\gamma_{\max}}{\gamma_{\min}} || \xi_{t_0} ||^2 + \gamma_{\max} l_1^{k} (||\xi_{t_0}||)  \\
 & -\frac{\gamma_{\max}}{\gamma_{\min}} \sum_{i=0}^{J-1} \left[ \beta^k ||\xi_{t_{n_0+iM}}^+||^2 - l_2(\xi_{t_{n_0+iM}}^+) \right].
\end{align*}
There exist $K>0$ and $\epsilon'>0$ such that for $x < \epsilon'$, we have $l_2(x) < \frac{\beta^k}{2} x$ and $ \gamma_{\max} l_1^{k} (x) < Kx$ (as $l_2(x) = O(x^2)$ and $l_1^k(x) = O(x^2)$) which gives:
$
|| \xi_{t_{n_{\max}}}^+ ||^2 \leq \left( \frac{\gamma_{\max}}{\gamma_{\min}} + K \right) || \xi_{t_0} ||^2.
$
Thus, for $|| \xi_{t_0} || < \frac{\epsilon'}{\sqrt{\frac{\gamma_{\max}}{\gamma_{\min}}+K}}$:
\begin{align}
|| \xi_{t_0+s} ||^2 \leq & \left( \frac{\gamma_{\max}}{\gamma_{\min}} + K + K \left( \frac{\gamma_{\max}}{\gamma_{\min}} + K \right) \right) || \xi_{t_0} ||^2 \nonumber \\
 & -\frac{\gamma_{\max}}{\gamma_{\min}}  \sum_{i=0}^{J-1} \frac{\beta^k}{2} ||\xi_{t_{n_0+iM}}^+||^2,
\label{eq::control_last formula}
\end{align}
which finally ensures
$$||\xi_{t_0}|| < \frac{1}{2} \epsilon' / \left( \frac{\gamma_{\max}}{\gamma_{\min}} + K + K \left( \frac{\gamma_{\max}}{\gamma_{\min}} + K \right) \right)\Rightarrow||\xi_{t_0+s}||\leq\epsilon'/2$$Reducing $\epsilon'$ if necessary to have $\epsilon'\leq\epsilon$, we have obtained $||\xi_{t_0+s}|| < \epsilon'$ for $s\in[0,t]\Rightarrow||\xi_{t_0+s}||\leq\epsilon'/2$ for sufficiently small $||\xi_{t_0}||$ (as Lemma \ref{lem::Lyapounov} applies).   Letting $t=\inf \left\lbrace s ,||\xi_{t_0+s}|| \geq \frac{3}{4}\epsilon'\right\rbrace$ for sufficiently small $||\xi_{t_0}||$ we end up with a contradiction if we suppose $t< + \infty$, which proves $t= + \infty$. All the previous results thus hold \emph{only} for sufficiently small $||\xi_{t_0}||$.

Moreover, \eqref{eq::control_last formula} shows that $\sum_{i=0}^{J-1} \frac{\beta^k}{2} ||\xi_{t_{n_0+iM}}^+||^2 $ is bounded and has positive terms thus $ || \xi_{t_{n_0+iM}}^+ ||^2$ goes to zero. Note also that $||P_t - \tilde{P}_t || \underset{t \rightarrow + \infty}{\longrightarrow} 0$ as a byproduct.

%%%%%%%%%%%%%%%%%%%%%%%%%%%%%%%%%%%%%%%%
% PROOF : CAR

\section{Proof of proposition \ref{prop::car}}
\label{proof::car}
Only conditions \eqref{cond::Psi} and \eqref{cond::observable} are non-trivial. Let $\Phi$ denote the flow of the dynamics. We have $$ \frac{d}{dt} \left[ (\Phi^t_{t_n})^T \Phi^t_{t_n} \right] = (\Phi^t_{t_n})^T \begin{pmatrix} 0 & 0 & -v_t \\ 0 & 0 & 0 \\ -v_t & 0 & 0 \end{pmatrix} \Phi^t_{t_n} \succeq -v_{\max} (\Phi^t_{t_n})^T \Phi^t_{t_n}  $$ as the eigenvalues of $ \begin{pmatrix} 0 & 0 & 1 \\ 0 & 0 & 0 \\ 1 & 0 & 0 \end{pmatrix}$ are $(1,-1,0)$. Thus, $\forall z\in \mathbb{R}^3, \frac{d}{dt} \log \left( z^T \left( \Phi^t_{t_n} \right)^T \Phi^t_{t_n} z   \right) \geqslant -v_{\max}$ and finally $z^T(\Phi^{t_{n+1}}_{t_n})^T \Phi_{t_n}^{t_{n+1}} z \geqslant \exp \left(- v_{\max} (t_{n+1}-t_n) \right) ||z||^2$ as $(\Phi^{t_n}_{t_n})^T \Phi_{t_n}^{t_n}=I_3$.  Thus $(\Phi^{t_{n+1}}_{t_n})^T \Phi_{t_n}^{t_{n+1}} \succeq \exp \left(- v_{\max} (t_{n+1}-t_n) \right)I_3$ and \eqref{cond::Psi} is verified. The difficult part of \eqref{cond::observable} is the lower bound. Denoting $Cov(V_n)$ by $N$ we will show:
\begin{align*}
  \exists \beta_1, & \forall n \in \NN, \quad \beta_1 I_3 \leq \hat{R}_{t_n}^T N^{-1} \hat{R}_{t_n} \\
 & + (\Phi_{t_n}^{t_{n+1}})^T \begin{pmatrix} 0_{1,2} & I_2 \end{pmatrix}^T \hat{R}_{t_{n-1}}^T N^{-1} \hat{R}_{t_{n-1}} \begin{pmatrix} 0_{1,2} & I_2 \end{pmatrix}  \Phi_{t_n}^{t_{n+1}}.
\end{align*}
That is to say that we want a lower bound on the quadratic form: 
\begin{align*}
M & \begin{pmatrix} \theta \\ u \end{pmatrix} =  \begin{pmatrix} \theta \\ u \end{pmatrix}^T \hat{R}_{t_n}^T N^{-1} \hat{R}_{t_n} \begin{pmatrix} \theta \\ u \end{pmatrix} \\
 &+ \begin{pmatrix} \theta \\ u \end{pmatrix}^T (\Phi_{t_n}^{t_{n+1}})^T \begin{pmatrix} 0_{1,2} & I_2 \end{pmatrix}^T \hat{R}_{t_{n-1}}^T N^{-1} \hat{R}_{t_{n-1}} \begin{pmatrix} 0_{1,2} & I_2 \end{pmatrix}  \Phi_{t_n}^{t_{n+1}} \begin{pmatrix} \theta \\ u \end{pmatrix} .
\end{align*}
We decompose $ \Phi_{t_n}^{t_{n+1}}$ as $ \Phi_{t_n}^{t_{n+1}} = \begin{pmatrix} 1 & \begin{matrix} 0 & 0 \end{matrix} \\ \delta V_n & T_n \end{pmatrix}$. To simplify the writing we introduce the norms $||x||_N^2 = x^T N^{-1} x$ and the associated scalar product $ \left< .,. \right>_N$. There exists $\alpha >0$ such that $\forall x\in \RR^2, ||x||_N \geq \alpha ||x||$. For any $\begin{pmatrix} \theta \\ u \end{pmatrix} \in \RR^3$ we have 
$
M \begin{pmatrix} \theta \\ u \end{pmatrix}  = ||\hat{R}_{t_n} u||_N^2 + || \theta \hat{R}_{t_{n-1}}\delta V_n + \hat{R}_{t_{n-1}} T_n u||_N^2 
  = ||\hat{R}_{t_n} u||_N^2 + \theta^2 ||\hat{R}_{t_{n-1}} \delta V_n||_N^2 + 2\theta \left< \hat{R}_{t_{n-1}} \delta V_n,\hat{R}_{t_{n-1}} T_n u \right>_N + ||\hat{R}_{t_{n-1}} T_n u||_N^2
$
 and for $\lambda \in ]0,1]$ we have:
\begin{align*}
M  \begin{pmatrix} \theta \\ u \end{pmatrix}
  = & ||\hat{R}_{t_n} u||_N^2 + (1-\lambda^2) \theta^2 || \hat{R}_{t_{n-1}}\delta V_n||_N^2 + \lambda^2 \theta^2 ||\hat{R}_{t_{n-1}}\delta V_n||_N^2  \\
 &+ 2\theta \left< \hat{R}_{t_{n-1}}\delta V_n, \hat{R}_{t_{n-1}} T_n u \right>_N + ||\hat{R}_{t_{n-1}} T_n u||_N^2 \\
  = & ||\hat{R}_{t_n} u||_N^2 + (1-\lambda^2) \theta^2 ||\hat{R}_{t_{n-1}}\delta V_n||_N^2 + ||\lambda \theta \hat{R}_{t_{n-1}}\delta V_n \\
 & + \frac{1}{\lambda} \hat{R}_{t_{n-1}} T_n u||_N^2 + (1-\frac{1}{\lambda^2}) ||\hat{R}_{t_{n-1}} T_n u||_N^2 \\
   \geq & \alpha \bigg[ ||\hat{R}_{t_n} u||^2 + (1-\lambda^2) \theta^2 ||\hat{R}_{t_{n-1}}\delta V_n||^2  \\
 &+ ||\lambda \theta \hat{R}_{t_{n-1}}\delta V_n + \frac{1}{\lambda} \hat{R}_{t_{n-1}} T_n u||^2 + (1-\frac{1}{\lambda^2}) ||\hat{R}_{t_{n-1}} T_n u||^2 \bigg] \\
     \geq & \alpha \left[ ||u||^2 + (1-\lambda^2) \theta^2 ||\delta V_n||^2 + (1-\frac{1}{\lambda^2}) ||u||^2 \right] \\
    \geq & \alpha \left[ ( 2-\frac{1}{\lambda^2})||u||^2 + \theta^2 + [\frac{1-\lambda^2}{2-\frac{1}{\lambda^2}} ||\delta V_n||^2-1] \theta^2 \right] \\
 \geq & \alpha (2-\frac{1}{\lambda^2}) \left( ||u||^2 + \theta^2 + [\frac{1-\lambda^2}{2-\frac{1}{\lambda^2}} v_{\min}^2-1] \theta^2 \right).
\end{align*}
As $ \frac{1-\lambda^2}{2-\frac{1}{\lambda^2}} \underset{\lambda \rightarrow \frac{1}{\sqrt{2}}^-}{\longrightarrow} +\infty $ there exists $\lambda_0$ such that:
$
{ M \begin{pmatrix}
\theta \\ u
\end{pmatrix} \geq \alpha \left( 2-\frac{1}{\lambda_0^2} \right) \left| \left| \begin{matrix}
\theta \\ u
\end{matrix} \right| \right|^2 }
$
and the result is true for $\beta_1 = \alpha (2-\frac{1}{\lambda_0^2})$.

\bibliographystyle{plain}
\bibliography{rhn}

\end{document}